\documentclass[a4paper,11pt]{article}
\usepackage{amssymb,amsmath,amsthm}
\usepackage[UKenglish]{babel}
\usepackage{microtype}
\usepackage{a4wide}
\usepackage{complexity}
\usepackage[dvipsnames]{xcolor}
\usepackage{hyperref} 
\hypersetup{colorlinks=true,citecolor=RoyalBlue,linkcolor=RoyalBlue,urlcolor=RoyalBlue}

\usepackage{graphicx}
\usepackage{cleveref}
\usepackage[shortlabels]{enumitem}

\usepackage[a4paper,top=1in,]{geometry}

\theoremstyle{plain}
\newtheorem{theorem}{Theorem}
\newtheorem{lemma}[theorem]{Lemma}
\newtheorem{corollary}[theorem]{Corollary}
\newtheorem{proposition}[theorem]{Proposition}
\theoremstyle{definition}
\newtheorem{remark}[theorem]{Remark}
\newtheorem{definition}[theorem]{Definition}

\newcommand{\NPhard}{NP-hard}
\newcommand{\NPhardness}{NP-hardness}
\newcommand{\barsuchthat}{|}

\newcommand{\tlint}{2Lin(2)}
\newcommand{\threelint}{3Lin(2)}
\newcommand{\threeSat}{3SAT}
\newcommand{\maxtlint}{Max-2Lin(2)}
\newcommand{\maxthreelint}{Max-3Lin(2)}
\newcommand{\maxcut}{Max-Cut}
\newcommand{\maxcsp}{Max-CSP}
\newcommand{\maxtsat}{Max-2Sat}
\newcommand{\maxflow}{Max-Flow}
\DeclareMathOperator{\maxflowval}{max\_flow}

\newcommand{\hadk}[1][k]{Had\textsubscript{$#1$}}
\newcommand{\maxhadk}[1][k]{Max-Had\textsubscript{$#1$}}

\newcommand{\hadkttlint}[1][k]{\hadk[#1]{}-to-\tlint{}}

\newcommand{\mincsp}{Min-CSP-deletion}
\newcommand{\mintlint}{Min-2Lin(2)-deletion}
\newcommand{\mincut}{Min-Cut}
\DeclareMathOperator{\mincutval}{min\_cut}

\DeclareMathOperator{\rs}{rs}
\DeclareMathOperator{\rsinf}{rs_\infty}
\DeclareMathOperator{\val}{val}
\DeclareMathOperator{\rsLP}{rsLP}
\DeclareMathOperator{\rsinfLP}{rs_\infty LP}
\DeclareMathOperator{\Sat}{Sat}
\DeclareMathOperator{\uniform}{uniform}

\DeclareMathOperator{\expect}{\mathbb{E}}
\newcommand{\nin}{\not\in}

\DeclareMathOperator{\Ffold}{\mathcal{F}_{\mathrm{fold}}}
\DeclareMathOperator{\Fk}{\mathcal{F}}
\newcommand{\bool}{\mathbb{F}_2}
\newcommand{\boolk}{\bool^k}
\newcommand{\boolkp}{\bool{}^{k'}}
\newcommand{\real}{\mathbb{R}}
\newcommand{\integers}{\mathbb{Z}}

\DeclareMathOperator{\affine}{affine}
\DeclareMathOperator{\dist}{dist}
\DeclareMathOperator{\Mk}{\mathcal{M}}
\DeclareMathOperator{\Nk}{\mathcal{N}}

\DeclareMathOperator{\fout}{out}
\DeclareMathOperator{\fin}{in}
\DeclareMathOperator{\leak}{leak}
\DeclareMathOperator{\lift}{lift}
\DeclareMathOperator{\dimension}{dim}

\begin{document}

\title{On the NP-Hardness Approximation Curve for Max-2Lin(2)}

\author{Bj{\"o}rn Martinsson\\
KTH\\
\texttt{bmart@kth.se}
}

\begin{titlepage}

\maketitle

\begin{abstract}
  In the \maxtlint{} problem you are given a system of equations on the form $x_i + x_j \equiv b \pmod{2}$, and your objective is to find an assignment that satisfies as many equations as possible. Let 
  $c \in [0.5, 1]$ denote the maximum fraction of satisfiable equations. In this paper we construct a curve $s (c)$ such that it is 
  \NPhard{} to find a solution satisfying at least a 
  fraction $s$ of equations. This curve either matches or improves all of the previously known inapproximability NP-hardness results for \maxtlint{}. In particular, we show that if $c \geqslant 0.9232$ then $\frac{1 - s (c)}{1 - c} > 1.48969$, which improves the NP-hardness inapproximability constant for the min deletion version of \maxtlint{}. Our work complements the work of O'Donnell and Wu that studied the same question assuming the Unique Games Conjecture.
  
  Similar to earlier inapproximability results for \maxtlint{}, we use a gadget reduction from the  $(2^k - 1)$-ary Hadamard predicate. Previous works used $k$ ranging from $2$ to $4$. Our main result is a procedure for taking a gadget for some fixed $k$, and use it as a building block to construct better and better gadgets as $k$ tends to infinity. Our method can
  be used to boost the result of both smaller gadgets created by hand $(k = 3)$ or larger gadgets constructed using a computer $(k = 4)$.
\end{abstract}

\end{titlepage}

\section{Introduction}
Maximum constraint satisfaction problems (\maxcsp{}s) form one of the most
fundamental classes of problems studied in computational complexity theory. A
\maxcsp{} is a type of problem where you are given a list of variables and a
list of constraints, and your goal is to find an assignment that satisfies as
many of the constraints as possible. Some common examples of \maxcsp{} are
\maxcut{} and \maxtsat{}. Every \maxcsp{} also has a corresponding \mincsp{}
problem where your objective is deleting as few constraints as possible to
make all of the remaining constraints satisfiable. The \mincsp{} problem is fundamentally the same optimisation problem as its corresponding \maxcsp{}, however their objective values are different.

\subsection{History of \maxcut{}}

The \maxcut{} problem is arguably both the simplest \maxcsp{} as well as the
simplest \NPhard{} problem. In the \maxcut{} problem you are given an undirected
graph, and your objective is to find a cut of the largest possible size. A cut
of an undirected graph is a partition of the vertices into two sets and the
size of a cut is the fraction of edges that connect the two sets relative to
the total number of edges. Solving \maxcut{} exactly is difficult, but there are
trivial approximation algorithms that get within a factor of $\frac{1}{2}$ of
the optimum. One such algorithm is randomly picking the cut by tossing one
coin per vertex.

Knowing this, one natural question is, how close can a polynomial time
algorithm get to the optimum? Goemans and Williamson partly answered this in a
huge breakthrough in 1995 {\cite{GW}} by applying semi-definite programming
(SDP) to create a polynomial time algorithm that finds a solution that is
within a factor of $\alpha_{\mathrm{GW}} \approx 0.87856$ of the optimum. At the
time, there was hope that Goemans and Williamson's algorithm could be improved
further to get even better approximation factors than $0.87856$, but no such
improvements were ever found. Instead, in 2004 Khot et al {\cite{KKMO}} proved
using the Unique Games Conjecture (UGC), that approximating \maxcut{} within a
factor of $\alpha_{\mathrm{GW}} + \varepsilon$ is \NPhard{} for any $\varepsilon >
0$. This conjecture had been introduced by Khot two years prior {\cite{Khot}}.
This was possibly the first result establishing the close connection between
UGC and SDP based algorithms.

To this day, UGC remains an open problem, and in particular no one has been
able to find an approximation algorithm for \maxcut{} with a better
approximation ratio than $\alpha_{\mathrm{GW}}$. In 2008 O'Donnell  and  Wu
{\cite{OW}} were able to very precisely describe the tight connection between
SDP based approximation algorithms for Max-Cut and UGC. They constructed a
curve $\mathrm{Gap}_{\mathrm{SDP}} (c) : [0.5, 1] \rightarrow [0.5, 1]$ with the
following two properties:

\begin{enumerate}
\item It is UGC-hard to find a cut of size $\mathrm{Gap}_{\mathrm{SDP}} (c) +
\varepsilon$ given that the optimal cut has size $c$ for any $\varepsilon >
0$. We here use \emph{UGC-hard} as a short hand for ``NP-hard under UGC''.

\item Within the $\mathrm{RPR}^2$-framework {\cite{rpr2,OW}}, there are polynomial time
algorithms that are guaranteed to find a cut of size at least
$\mathrm{Gap}_{\mathrm{SDP}} (c) - \varepsilon$ if the optimal cut has size $c$. The
$\mathrm{RPR}^2$-framework is a generalisation of Goemans and~Williamson's
algorithm.
\end{enumerate}

This means that their work both describe the best known polynomial time
approximation algorithms for Max-Cut, and also show that under UGC these
approximation algorithms cannot be improved. It is important to note that
their algorithmic results do not require UGC. We emphasis that one implication
of their result is that giving efficient algorithms with a better performance
would disprove UGC.

\subsection{NP-hardness inapproximability of \maxtlint{}}

\maxtlint{} is a \maxcsp{} that is very closely related to \maxcut{}. An instance
of \maxtlint{} is a system of linear equations on the form $x_i + x_j \equiv b
\pmod{2}$, and the objective is to find an assignment that satisfies as many
equations as possible. \maxcut{} is the special case where we only allow
equations with right hand side equal to $1$. This implies that any hardness
result for \maxcut{} immediately yields the same hardness result for
\maxtlint{}. One example of this is the UGC-hardness of \maxcut{} described by
the $\mathrm{Gap}_{\mathrm{SDP}} (c)$ curve by O'Donnell  and  Wu {\cite{OW}}.

Furthermore, O'Donnell  and  Wu's algorithmic results {\cite{OW}} also
directly carries over to \maxtlint{}. This is because the
$\mathrm{RPR}^2$-framework that they relied on uses odd rounding functions, and
therefore does not differentiate between \maxcut{} and \maxtlint{}.

The conclusion is that the $\mathrm{Gap}_{\mathrm{SDP}} (c)$ describes a tight
connection between the UGC-hardness of \maxtlint{} as well as the best known
polynomial time approximation algorithms for \maxtlint{}. On the other hand,
the NP-hardness inapproximability of \maxtlint{} is not well understood. The
strongest \NPhardness{} inapproximability results known for \maxtlint{}
({\cite{Has}}, {\cite{Wiman}}) are still far off from the UGC-hardness
described by the $\mathrm{Gap}_{\mathrm{SDP}} (c)$ curve.

The aim of this paper is to improve the state of the art \NPhardness{}
inapproximability of \maxtlint{} and also to give the full
picture of the state of the art \NPhardness{} inapproximability of
\maxtlint{}. We do this by constructing a curve $s(c) :
[0.5, 1] \rightarrow [0.5, 1]$ such that it is \NPhard{} to distinguish between
instances where the optimal assignment satisfies a fraction of $c$ of the
equations, and instances where all assignments satisfy at most a fraction of
$s(c)$ of the equations. Our curve either matches or improves all previously
known \NPhardness{} inapproximability results for \maxtlint{}. We construct the curve by solving a separate optimisation problem for
each value of $c$, so our result covers the entire spectrum of $c \in [0.5,1]$.

Our result complements the work by O'Donnell  and  Wu {\cite{OW}}. Our curve
describes the state of the art \NPhardness{} inapproximability of \maxtlint{}
while O'Donnell  and  Wu's $\mathrm{Gap}_{\mathrm{SDP}} (c)$ curve describes the
UGC-hardness of \maxtlint{}. It is worth noting that UGC is still an open
problem that over the years has been the subject of much debate. There are
results that indicate that UGC might be true, such as the proof of the closely
related $2$-to-$2$ Games Conjecture {\cite{2to2}}. But on the other hand there
are also results that indicate the UGC might be false, such as the existence
of subexponential algorithms for Unique Games {\cite{subexp}}. Currently there
is no consensus for whether UGC is true or not. It is for this reason that it
is important to study \NPhardness{} independent of UGC, especially for
fundamental problems such as \maxtlint{}.

\subsection{Gadget reductions}

Gadgets are the main tools used to create reductions from one \maxcsp{} $\Phi$
to another \maxcsp{} $\Psi$. A gadget is a description of how to translate a
specific constraint $\varphi$ of $\Phi$ into one or more constraints of
$\Psi$. For example, if $\Phi$ is \maxthreelint{} and $\Psi$ is \maxcut{}, then a
gadget from $\varphi$ to $\Psi$ is a graph. A gadget is allowed to use both
the original variables in the constraint $\varphi$, which are called
\emph{primary variables}, and new variables specific to the gadget, which
are called \emph{auxiliary variables}.

The standard technique used to construct gadgets is to follow the ``automated
gadget'' framework of Trevisan et al {\cite{Trev}}. This framework describes
how to construct a gadget by solving a linear program and also proves that the constructed gadget is optimal. This framework is mainly used to construct
gadgets for small and simple \maxcsp{}s. This is because the number of variables
in the gadget scales exponentially with the number of satisfying assignments
of $\varphi$. Furthermore, the number of constraints in the LP scales
exponentially with the number of variables, so it scales double exponentially
with the number of satisfying assignment of $\varphi$.

As an example let us take the gadget from \maxthreelint{} to \maxtlint{} used
by Håstad {\cite{Has}}, which was originally constructed by Trevisan et
al {\cite{Trev}}. A constraint in \maxthreelint{} has $4$ satisfiable assignments.
Having $4$ satisfiable assignments means that the gadget uses $2^4 = 16$
variables. Furthermore, since \maxtlint{} allow negations, half of these
variables can be removed because of negations. So the actual number of
variables in the gadget is $2^{4 - 1} = 8$. This in turn implies that the
number of constraints in the LP is $2^8 = 256$. This number is small enough
that it is feasible for a computer to solve the LP. In this paper we are
interested in constructing gadgets from generalisations of \maxthreelint{}, called
the Hadamard \maxcsp{}s. These have significantly more satisfying assignments
than \maxthreelint{}. It is easy to see that a simple-minded application of the
``automated gadget'' framework leads to an LP that is far too large to naively
be solved by a computer. This means that we have to deviate from the
``automated gadget'' framework in order to construct our gadgets.

Gadgets have two important properties, called \emph{soundness $s$} and \emph{completeness
$c$}. If a gadget is constructed using the ``automated gadget'' framework, then
it is trivial to calculate the completeness of the gadget. On the other hand,
calculating the soundness of a gadget from $\Phi$ to $\Psi$ involves solving
instances of $\Psi$. In practice, calculating the soundness of a large gadget
can be very difficult since $\Psi$ is usually an NP-hard problem.

Gadgets can be constructed with different goals in mind. The case that we are
interested in is finding the gadget with the largest soundness for a fixed
completeness. This is what allows us to construct our curve $s(c)$. In
general there are also other objectives that could be of interest when
constructing gadgets. One such case is finding the gadget with the smallest
ratio of $\frac{s}{c}$. This corresponds to finding the best lower bound for
the approximation ratio of \maxtlint{}. Another possibility is to maximise
$\frac{1 - s}{1 - c}$. This corresponds to finding the best upper bound for
the approximation ratio of \mintlint{}. It is possible to use the
``automated gadget'' framework by Trevisan et al {\cite{Trev}} to find the
optimal gadgets for all of these scenarios.

\subsection{The Hadamard \texorpdfstring{\maxcsp{}s \maxhadk{}}{Max-CSPs Max-Hadk}}

One of the earliest gadget reductions used to show \NPhardness{}
inapproximability of \maxtlint{} is a gadget reduction from
\maxthreelint{} used by Håstad in his classical paper from 1997
{\cite{Has}}, which was constructed by Trevisan et al {\cite{Trev}}.
More recently, \NPhardness{} inapproximability results for \maxtlint{} have
used gadget reductions from a generalisation of \maxthreelint{} called
the Hadamard \maxcsp{}s {\cite{Has2,Wiman}}. The $(2^k - 1)$-ary
Hadamard \maxcsp{}, $k \geqslant 2$, is a constraint satisfaction problem where
a clause is satisfied if and only if its literals form the truth table of a
linear $k$-bit Boolean function. The $(2^k - 1)$-ary Hadamard CSP is denoted
by \maxhadk{}. One special case is $k = 2$, where the number of
literals of a clause is $3$. It turns out that this case coincides with
\maxthreelint{}. This means that \maxhadk{} can be seen as a
generalisation of \maxthreelint{}.

There are mainly two reasons as to why \maxhadk{} is useful for gadget
reductions. The first reason is that \maxhadk{} is a very sparse CSP.
It being sparse refers to the number of satisfiable assignments of a clause
being few in relation to the total number of possible assignments. The number
of satisfying assignments of a clause is just $2^k$, one for each linear
$k$-bit Boolean function, while the total number of possible assignments is
$2^{(2^k - 1)}$.

The second reason is that \maxhadk{} is a \emph{useless
predicate} for any $k \geqslant 2$, which is an even stronger property than
being approximation resistant. This was shown by Chan in 2013 {\cite{Chan}}.
\maxhadk{} being a useless predicate means that if you are given a
nearly satisfiable instance of \maxhadk{}, then it is \NPhard{} to find an
assignment such that the distribution over the $(2^k - 1)$ long bit strings
given by the literals of the clauses is discernibly different from the uniform
distribution.

\subsubsection{Historical overview of \texorpdfstring{\hadkttlint{}}{Hadk-to-2Lin(2)} gadgets}

In 1996, Trevisan et al {\cite{Trev}} constructed the optimal gadget from \maxhadk[2]{} to \maxtlint{}. They showed that the \maxhadk[2]{}
gadget that minimises $\frac{s}{c}$ is the same gadget as the one that
maximises $\frac{1 - s}{1 - c}$. Furthermore, since this gadget is very small,
using only $8$ variables, they were able to construct it using the ``automated
gadget'' framework.

In 2015, Håstad et al. {\cite{Has2}} constructed gadgets from
\maxhadk[3]{} to \maxtlint{}. They showed that the \maxhadk[3]{}
gadget that minimises $\frac{s}{c}$ is equivalent to the \maxhadk[2]{}
gadget. So using \maxhadk[3]{} over \maxhadk[2]{} does not give an
improved hardness for the approximation ratio of \maxtlint{}. However, the
\maxhadk[3]{} gadget that maximises $\frac{1 - s}{1 - c}$ is notably
better than the \maxhadk[2]{} gadget. This gadget is relatively small,
only using 128 variables. This is too many variables for it to be possible to
naively apply the ``automated gadget'' framework. However, Håstad et
al. were still able to construct and analyse the optimal gadget by hand based
on ideas from the ``automated gadget'' framework.

In 2018, Wiman {\cite{Wiman}} constructed gadgets from \maxhadk[4]{} to
\maxtlint{}. Note that \maxhadk[4]{} gadgets have $2^{15}$ variables.
Calculating the soundness of a such a gadget requires solving an instance of
\maxtlint{} with $2^{15}$ variables, which is infeasible to do by hand or even
with a computer. Wiman initially followed the ``automated gadget'' framework.
However, in order to be able to calculate the soundness of the gadget, Wiman
relaxed the \maxtlint{} problem into a \maxflow{} problem. This
\emph{relaxed soundness} $\rs$ is an upper bound of the true soundness.
This relaxation made it possible for Wiman to use a computer to find the
gadget that maximises $\frac{1 - \rs}{1 - c}$. Wiman's relaxation was
successful, since by using it he was able to find a \maxhadk[4]{} gadget
that was better than the optimal \maxhadk[3]{} gadget. Note, however,
that by using a relaxation, it is uncertain whether Wiman found the optimal
\maxhadk[4]{} gadget or not.

\subsubsection{Our \texorpdfstring{\hadkttlint{}}{Hadk-to-2Lin(2)} gadgets}

In this paper, we construct gadgets from \maxhadk{} to \maxtlint{} for
$k$ approaching infinity. Recall that a gadget uses $2^{2^k - 1}$ variables,
so using a computer to construct gadgets for $k > 5$ is normally impossible.
We get around this limitation by introducing a procedure for taking
\maxhadk{} gadgets and transforming them to \maxhadk[k']{}
gadgets, for $k' > k$. We refer to this procedure as the \emph{lifting} of
a \maxhadk{} gadget into a \maxhadk[k']{} gadget. Two of the
properties of lifting is that the completeness stays the same and the
soundness does not decrease.

To show \NPhardness{} of approximating \maxtlint{}, we start by constructing
\maxhadk{} to \maxtlint{} gadgets for $k = 4$ using a computer. We
then analytically prove an upper bound of Wiman's relaxed soundness of the
lifting of these gadgets as $k' \rightarrow \infty$.

The method we use to construct our gadgets is by solving an LP. This LP is
similar to what Wiman could have used to construct his gadget. The difference
is that the LP we use is made to minimise the soundness of the lifted gadget,
instead of minimising the soundness of the gadget itself. If done naively,
this LP would have roughly $2^{3 \cdot (2^k - 1)} = 2^{45}$ variables. But by
making heavy use of symmetries of the LP, we are able to bring it down to a
feasible size.

The main technical work of this paper is proving an upper bound on Wiman's
relaxed soundness of a lifted gadget as $k' \rightarrow \infty$. Recall that
calculating Wiman's relaxed soundness involves solving instances of \maxflow{}.
As $k'$ tends to infinity, the size of these instances also tend to infinity.
In order to lower bound the value of these \maxflow{} problems, we introduce the
concept of a type of infeasible flows which we call \emph{leaky flows}. A leaky flow is a flow for which the conservation of flows constraint has been
relaxed. This allows leaky flows to attain higher values compared to feasible
flows. We then show that by randomly overlapping leaky flows onto the large
\maxflow{} instances, we are able to get closer and closer to a feasible flow as
the size of the instances tend to infinity.

\subsection{Our results and comparison to previous results}

Using a gadget reductions from \maxhadk{} to \maxtlint{}, we are able
to construct a curve $s (c) : [0.5, 1] \rightarrow [0.5, 1]$ such that it is
\NPhard to distinguish between instances of \maxtlint{} where the optimal
assignment satisfies a fraction of $c$ of the equations and instances where
all assignments satisfy at most a fraction of $s(c)$ of the equations. This curve does not have an explicit formula. Instead, each point on the curve is defined as the solution to an LP, which we solve using a computer.

\begin{theorem}
  \label{main_res}Let $s(c) : [0.5, 1] \rightarrow [0.5, 1]$ be the curve defined in Definition \ref{curve}. Then for every sufficiently small $\varepsilon > 0$, it is NP-hard to distinguish between instances of \maxtlint{} such that
  \begin{description}
  \item[Completeness] There exists an assignment that satisfies a fraction at
  least $c - \varepsilon$ of the constraints.
  \item[Soundness] All assignments satisfy at most a fraction $s(c) + \varepsilon$
  of the constraints.
  \end{description}
\end{theorem}

A notable point on the curve is $c = \frac{590174949}{639271832} \approx
0.9232$ and $s (c) = \frac{141533171}{159817958} \approx 0.8856$. This is the
point on the curve that gives the highest \NPhardness{} inapproximability factor
$\frac{1 - s}{1 - c}$ of \mintlint{}.

\begin{corollary}
  It is \NPhard{} to approximate \mintlint{} within a factor of
  $\frac{73139148}{49096883} + \varepsilon \approx 1.48969 + \varepsilon$.
\end{corollary}

In order to be able to compare our curve $s (c)$ to prior results, we plot our
curve together with O'Donnell and Wu's $\mathrm{Gap}_{\mathrm{SDP}} (c)$ curve
{\cite{OW}}, which, as discussed earlier, describes both the UGC-hardness of
\maxtlint{}, as well as the best known polynomial time approximation
algorithms of \maxtlint{}. Additionally, we also include historical
NP-hardness inapproximability results as points in the diagram. We have also
marked the point $(c, s)$ where Goemans and Williamson's \ algorithm achieves
the approximation ratio of $\frac{s}{c} = \alpha_{\mathrm{GW}} \approx 0.87856$.
This point was shown to be UGC-hard by Khot et al. in 2004 {\cite{KKMO}}.

\begin{figure}
  \centering
  {\includegraphics[width=0.85\textwidth]{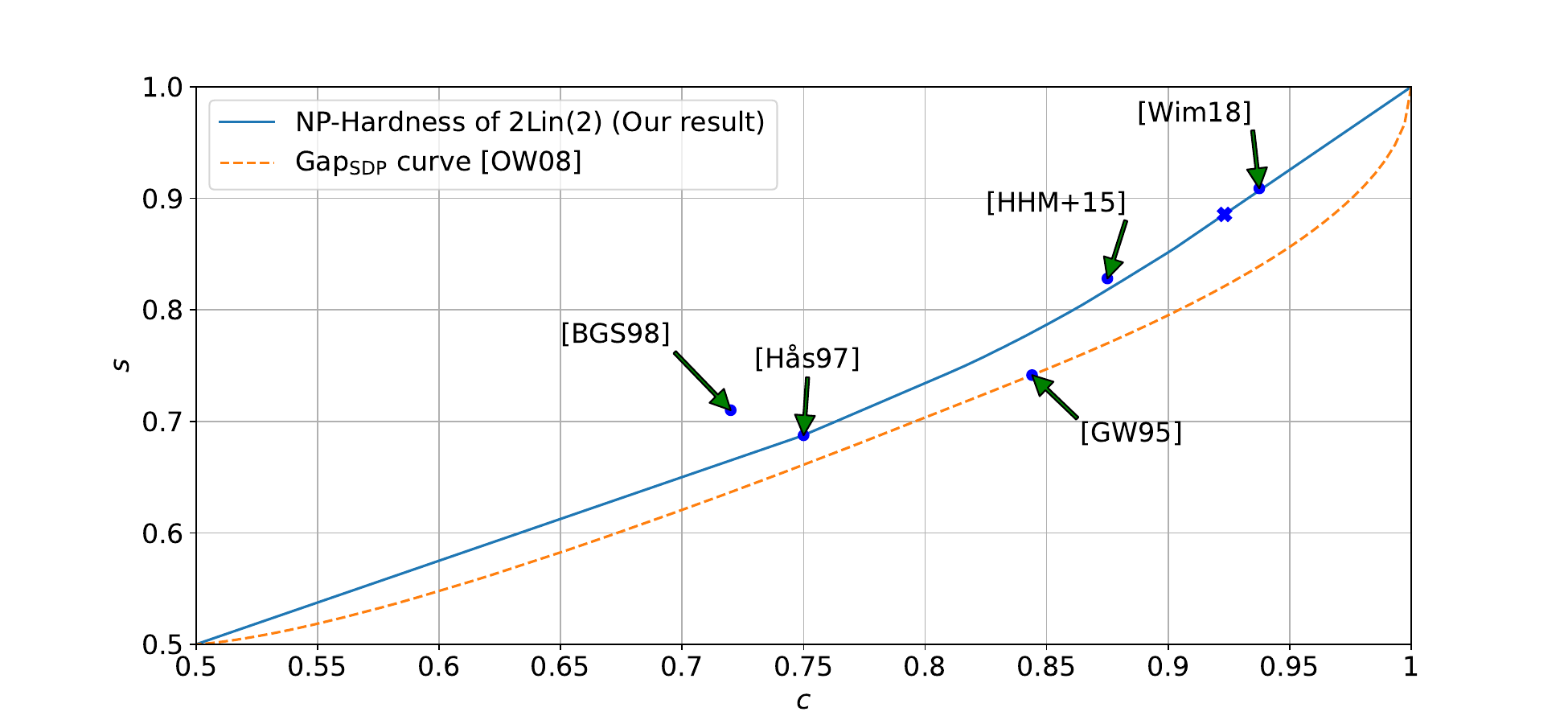}}
  \caption{\label{fig1}The $y$-axis shows the soundness $s$ and the $x$-axis
  the completeness $c$. The blue filled curve is our \NPhardness{} curve $s
  (c)$. The red dashed curve is the $\mathrm{Gap}_{\mathrm{SDP}} (c)$ by O'Donnell
  and Wu's {\cite{OW}}. The points marked with arrows are prior
  inapproximability results of \maxtlint{}. The blue cross on the curve
  marks our best inapproximability result for \mintlint{}, see Figure
  \ref{fig3}. Note that both of the curves in this figure are convex
  functions.}
  {\includegraphics[width=0.85\textwidth]{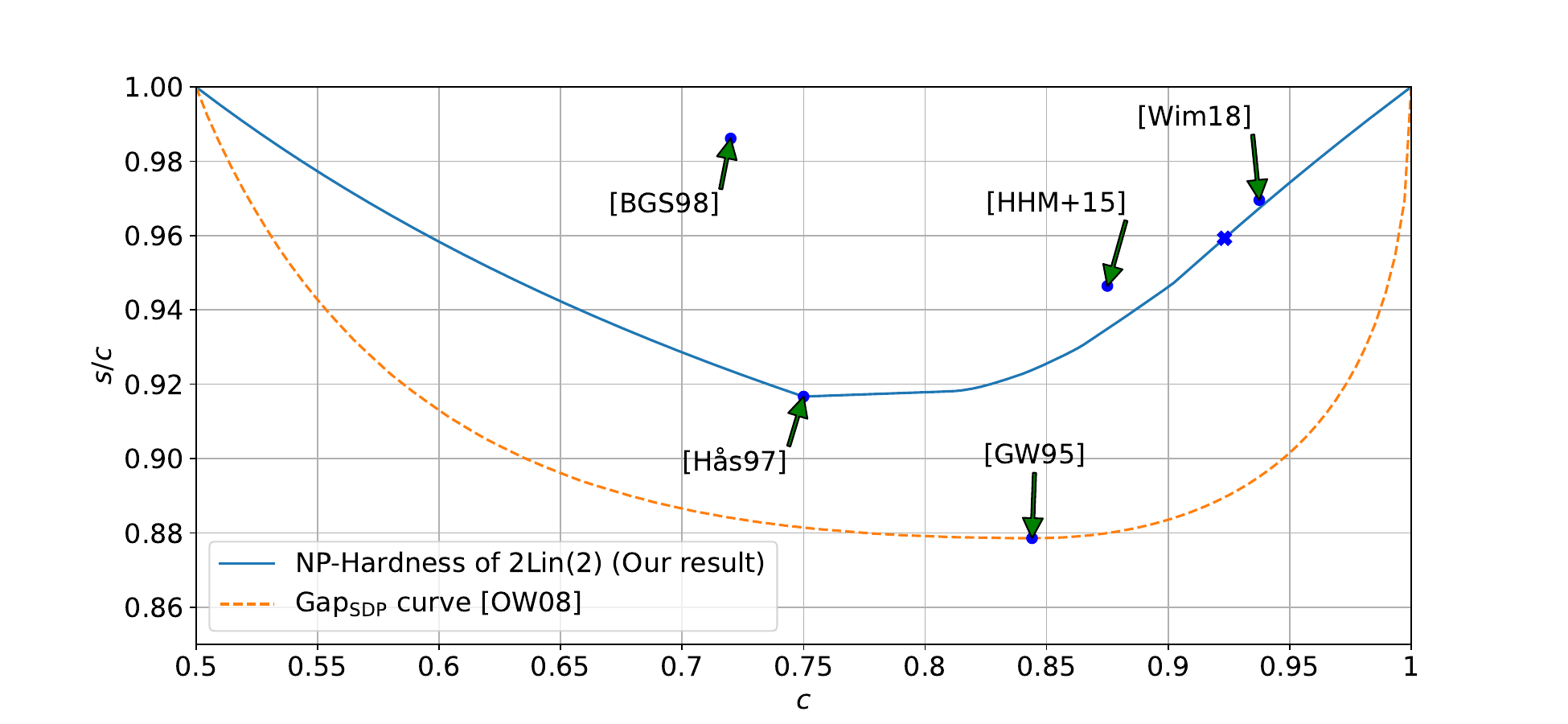}}
  \caption{\label{fig2}The $y$-axis shows $s / c$, which corresponds to the
  approximation ratio of \maxtlint{}. The point on the curve $c (s)$ that
  minimises this ratio is $c = 3 / 4$ and $s (c) = 11 / 16$, which exactly
  matches Håstad's result from 1997 {\cite{Has}}.}
  {\includegraphics[width=0.85\textwidth]{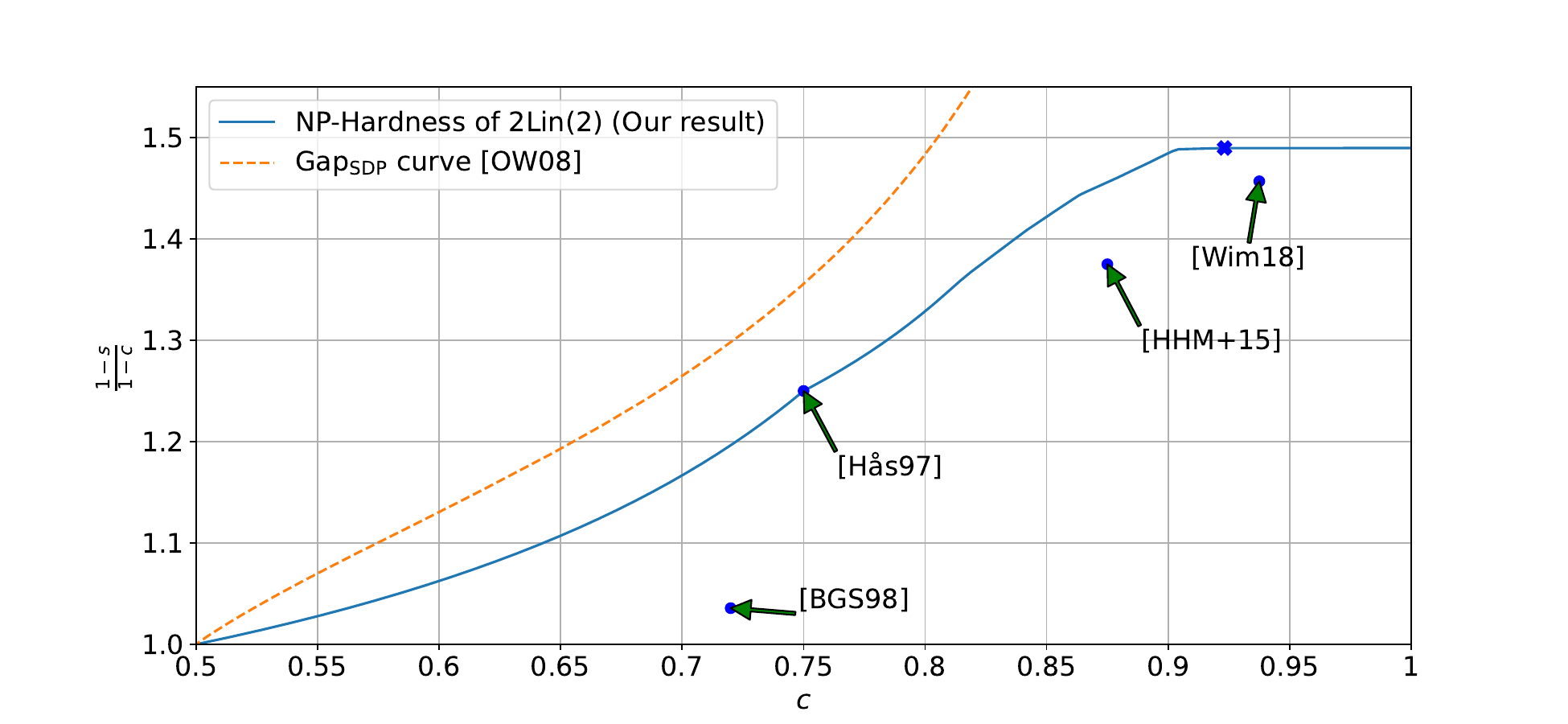}}
  \caption{\label{fig3}The $y$-axis shows $(1 - s) / (1 - c)$, which
  corresponds to the approximation ratio of \mintlint{}. This ratio
  reaches its maximum $\frac{1 - s (c)}{1 - c} = \frac{73139148}{49096883}
  \approx 1.4896$ at $c = \frac{590174949}{639271832}$, which is marked by a
  blue cross. The curve stays constant after this point.}
\end{figure}

The curve $s(c)$ is plotted in three Figures. All three Figures contain the
same exact same data, but the data is plotted in different ways. In Figure \ref{fig1} the soundness $s(c)$ is on the $y$-axis and the completeness $c$
is on the $x$-axis. This plot has the disadvantage that to the eye, it is
difficult to distinguish the exact shape of the curve $s (c)$. In the next
plot, Figure \ref{fig2}, $\frac{s (c)}{c}$ is on the $y$-axis and $c$ is on
the $x$-axis. This plot describes the approximation ratio of
\maxtlint{}. The third plot, in Figure \ref{fig3}, has $\frac{1 - s
(c)}{1 - c}$ on the $y$ axis and $c$ on the $x$-axis. This plot describes the
approximation ratio of \mintlint{}.

It is important to note that the curves in Figure \ref{fig1} are convex functions since it is possible to take the convex combination of two hard
instances using disjoint sets of variables. One implication from this is that
it is possible to construct NP-hardness curves using any of the points $(c,
s)$ by drawing two lines, one from $(0.5, 0.5)$ to $(c, s)$ and one from $(c,
s)$ to $(1, 1)$. This means that all of the historical inapproximability results can also be described using convex curves.

In Figures $\ref{fig1}$-\ref{fig3} prior inapproximability results for
\maxtlint{} are marked as dots. Bellare et al {\cite{BGS}} was first to give
an explicit NP-hardness result, which had $c = 0.72$ and $s = 0.71$. In 2015,
Håstad et al {\cite{Has2}} used Chan's result {\cite{Chan}} to create a
gadget reduction from \maxhadk[3]{} which had $c = \frac{7}{8}$ and $s =
\frac{53}{64}$. This result became the new record for the upper bound of the
approximation ratio of \mintlint{}, as seen in Figure
\ref{fig3}. Three years later, Wiman {\cite{Wiman}} further improved on this
result by using \maxhadk[4]{} instead of \maxhadk[3]{}. Wiman's
\maxhadk[4]{} gadget has $c = \frac{15}{16}$ and $s =
\frac{3308625759}{3640066048} \approx 0.9089$. This further improved the upper
bound on the approximation ratio of \mintlint{}.

Similar to earlier results, the technique we use to construct our curve is
also a gadget reduction from \maxhadk{} to \maxtlint{}. But
instead of using a gadget reduction from \maxhadk{} for a fixed $k$,
we instead let $k$ tend to infinity. This improves the quality of our gadget.
One example of such an improvement is our upper bound on the approximation
ratio of \mintlint{}, which can be seen in Figure \ref{fig3}. The
ratio $\frac{1 - s (c)}{1 - c}$ is maximised on our curve at $c =
\frac{590174949}{639271832} \approx 0.9232$ and $s =$
$\frac{141533171}{159817958} \approx 0.8856$, which is marked by a blue cross
in Figure \ref{fig3}.

\subsection{The limitations of \texorpdfstring{\hadkttlint{}}{Hadk-to-2Lin(2)} gadget
reductions}

In Figure \ref{fig1}, it is possible to see a clear gap between our $s(c)$
curve and O'Donnell and Wu's $\mathrm{Gap}_{\mathrm{SDP}}(c)$ curve {\cite{OW}}.
The gap is especially noticeable in Figure \ref{fig3}, since the behaviour of
the two curves are completely different when $c$ is close to $1$. One natural
question is, how close can a \hadkttlint{} gadget
reduction get to the $\mathrm{Gap}_{\mathrm{SDP}} (c)$ curve?

Håstad et al {\cite{Has2}} showed that any gadget reduction from a
Hadamard \maxcsp{} to \maxtlint{} can never achieve an approximation ratio for \mintlint{} better than $\frac{1}{1 - e^{- 0.5}} \approx 2.54$. In
Appendix \ref{appendix:B} we show that any gadget reduction from a Hadamard CSP
to \maxtlint{} that uses Wiman's soundness relaxation can never achieve an
approximation ratio of \mintlint{} better than $2$. Both $2.54$ and
$2$ are fairly large in comparison to the current best value of $1.48969$
shown in Figure \ref{fig3}. So it is potentially possible to still improve our
results in the future using a \hadkttlint{} gadget
reduction for some $k \geqslant 4$. However, these limitations means that it
is impossible to make $s(c)$ match the behaviour of $\mathrm{Gap}_{\mathrm{SDP}}
(c)$ when $c$ is close to $1$.

\subsection{Outline of proof}

Our result is based on \hadkttlint{} gadget reduction
for arbitrary large values of $k$. We start from the ``automated gadget''
framework by Trevisan et al {\cite{Trev}}. In this framework, computing the
soundness of of a \hadkttlint{} gadget involves solving
a \maxtlint{} problem. Following the work of Wiman {\cite{Wiman}},
we relax the soundness computation to a \maxflow{} problem on the
$2^k$-dimensional hypercube. Using symmetries, it is computationally feasible
to construct \hadkttlint{} gadgets that are optimal
with respect to the relaxed soundness for $k \leqslant 4$.

In order to be able harness the power of arbitrarily large $k$, we define a
procedure of embedding a \hadkttlint{} gadget $G$
inside a \hadkttlint[k']{} gadget where $k' > k$. By
overlapping multiple different embeddings of $G$, we construct a gadget $G'$
for an arbitrarily large $k'$.

Recall that the relaxed soundness computation is a \maxflow{} problem, which
can be expressed as an LP. By carefully relaxing this LP, we are able to
create an infeasible flow solution to $\rs(G)$, such that if we lift
it, it becomes an almost feasible flow of $\rs(G')$. The underlying
idea for this relaxation is based on leaky flows (flows where the flow
entering a node can be different than the flow exiting the node). The
``leaks'' of a leaky flow are signed, so random overlap of leaky flows can
result in a feasible flow. We show that this is actually the case for the
solution to our relaxed LP using a second order moment analysis.

The final step is to construct the \hadkttlint{} gadget
$G$ and its corresponding leaky flow for $k = 4$ used in the embedding. This
construction is naturally done using a rational LP solver to solve the relaxed
LP.

\subsection{Organisation of paper}

Section \ref{sec2} contain the preliminaries. It introduces \maxcsp{}s and the automated gadget framework. Section
\ref{sec_gadget} introduces Wiman's relaxed soundness and the infinity relaxed
soundness in terms of an LP. This section also states our main Lemma, Lemma
\ref{thm_important}, relating the infinity relaxed soundness to the relaxed
soundness. Appendix \ref{2.3} is about Max-Flow, and it proves some general theorems about how symmetries can be used to simplify Max-Flow problems. Appendix \ref{appendix:B} contain an analysis of relaxed soundness,
and how it relates to the (true) soundness. Appendix \ref{sec:affine} studies
affine maps. These affine maps are used both to analyse the infinity relaxed
soundness, and to describe the symmetries of the LPs. Appendix \ref{sec_proof}
contains the proof of Lemma \ref{thm_important} using the affine maps.
Appendix \ref{sec_comp} describes the procedure we use for constructing and
verifying the gadgets. Section \ref{sec:num} contains our numerical
results. This includes both plots and tables of various \hadkttlint{} gadgets. Finally, Appendix \ref{sec:used} contains a compact description of all gadgets that we construct.

\section{Preliminaries\label{sec2}}

This section is split into three parts. In Subsection \ref{2.1} we introduce
some basic concepts and notations for Boolean functions $\boolk \rightarrow \{
1, - 1 \}$. After that, in Subsection \ref{2.2} we formally define the $(2^k -
1)$-ary Hadamard predicate. The last subsection,
Subsection \ref{sec:auto}, introduces the ``automated gadget'' framework by
Trevisan et al {\cite{Trev}}, and explains the classical result of how to
construct reductions from the $(2^k - 1)$-ary Hadamard predicate to \maxtlint{}.

\subsection{Boolean functions}\label{2.1}

A $k$-bit Boolean function is a function that takes in $k$ bits and outputs
one bit. The $k$ input bits should be thought of as a vector in a
$k$-dimensional vector field over $\bool{}$. On the other hand, the output bit
is a scalar. For convenience, we denote the vectors as being elements in
$\boolk{}$ and the scalars as elements in $\real{}$, where a scalar bit is
represented as $1$ (False) or $- 1$ (True).

\begin{definition}
  The set of $k$-bit Boolean functions is denoted by $\Fk_k = \left\{
  f : \boolk{} \rightarrow \{ 1, - 1 \} \right\}$.
\end{definition}

One special type of Boolean functions that is of great importance is the set
of linear Boolean functions. Each linear Boolean function in $\boolk{}$ corresponds to an element $\alpha \in \boolk{}$, and is denoted by
$\chi_{\alpha}$.

\begin{definition}
  For $\alpha \in \boolk{}$ let $\chi_{\alpha} \in \Fk_k$ be denote the
  function
  \begin{eqnarray*}
    \chi_{\alpha} (x) & = & (- 1)^{(\alpha, x)}
  \end{eqnarray*}
  where $(\alpha, x) = \sum_{i = 1}^k \alpha_i x_i  \pmod{2}$.
\end{definition}

Any Boolean function can be represented as a sum of linear Boolean functions
using the Fourier transform.

\begin{proposition}
  \label{prop:finv}Given $f \in \Fk_k$, then
  \begin{eqnarray*}
    f (x) & = & \sum_{\alpha \in \boolk{}} \chi_{\alpha} (x)  \hat{f}_{\alpha},
  \end{eqnarray*}
  where $\hat{f}_{\alpha}$ denotes the Fourier transform of $f$ at $\alpha$,
  defined as
  \begin{eqnarray*}
    \hat{f}_{\alpha} & = & \frac{1}{2^k} \sum_{x \in \boolk{}} \chi_{\alpha} (x)
    f (x), \quad \alpha \in \boolk{} .
  \end{eqnarray*}
\end{proposition}

The Fourier transform is used to define the supporting affine subspace of a
Boolean function. This also gives a natural definition for the dimension of a
Boolean function.

\begin{definition}
  Given $f \in \Fk_k$, its supporting affine sub-space $\affine
  (f)$ is the affine span of $\left\{ \alpha \in \boolk{} : \hat{f}_{\alpha}
  \neq 0 \right\}$.
\end{definition}

\begin{definition}
  Let $\dimension (f), f \in \Fk_k$, denote the dimension of $\affine
  (f)$.
\end{definition}

\begin{remark}
  Affine functions have dimension $0$.
\end{remark}

The distance between two Boolean function is given by the normalised Hamming
distance.

\begin{definition}
  Let $\dist : \Fk_k \times \Fk_k \rightarrow
  \real$ be the normalised Hamming distance between two Boolean
  functions, i.e.
  \begin{eqnarray*}
    \dist (f_1, f_2) & = & \frac{1}{2^k} \sum_{x \in \boolk{}} \frac{1 -
    f_1 (x) f_2 (x)}{2} .
  \end{eqnarray*}
\end{definition}

\subsection{\maxcsp{}}\label{2.2}

This section introduces Constraint Satisfaction Problems (CSP) and
\maxcsp{}. The framework we use is that CSPs are defined by predicates,
which describe which kind of constraints that can appear in the CSP.

\begin{definition}
  An $m$-ary predicate is a function $\phi \in \Fk_m$. The predicate
  is said to be satisfied by $x \in \Fk_2^m$ if $\phi (x) = - 1$.
  Otherwise $x$ is said to violate $\phi$. The set of $x \in \Fk_2^m$
  that satisfies $\phi$ is denoted by $\Sat (\phi)$.
\end{definition}

Given a set of Boolean variables $V$ and a $m$-ary predicate $\phi$, a
$\phi$-\emph{constraint} $\mathcal{C}$  is a tuple $((x_1, b_1), \ldots, (x_m, b_m))$
where $x_i \in V, i = [m]$, and $b_i \in \bool$, $i \in [m]$, where all of the
$x_i$'s are distinct. The constraint $\mathcal{C}$ is said to be satisfied if
\begin{eqnarray*}
  \phi (b_1 + x_1, \ldots, b_n + x_m) & = & - 1,
\end{eqnarray*}
where $+$ denotes the xor-operation. In other words, if $b_i = 1$ then $x_i$
is negated.

\begin{definition}
  Given a $m$-ary predicate $\phi$, an instance $\mathcal{I}$ of the
  Max-$\phi$-CSP is a variable set $V$ and a distribution of $\phi$-constraints
  over $V$. The Max-$\phi$-CSP optimisation problem is; given an instance
  $\mathcal{I}$, find the assignment $A : V \rightarrow \bool$ that maximises
  the fraction of satisfied constraints in $\mathcal{I}$. The optimum is
  called the value of $\mathcal{I}$.
\end{definition}

The main \maxcsp{}s of interest in this paper are the Hadamard
\hadk{} \maxcsp{}, and \maxtlint{} and \maxthreelint{}. These have
the following predicates.

\begin{definition}
  The \tlint{} predicate is the function $f (x, y) = (- 1)^{x + y +
  1}$. Similarly, the \threelint{} predicate is the function $f (x, y,
  z) = (- 1)^{x + y + z + 1}$.
\end{definition}

\begin{definition}
  The Hadamard \hadk{} predicate for $k \geqslant 2$ is a $(2^k -
  1)$-ary predicate. There is one input variable per non-empty subset $S
  \subseteq [k]$. The \hadk{} predicate is satisfied by a binary input
  string $\{ x_S \}_{\varnothing \neq S \subseteq [k]}$ if and only if there
  exists some $\beta \subseteq [k]$ such that
  \begin{eqnarray*}
    \chi_{\beta} (S) & = & (- 1)^{x_S} 
  \end{eqnarray*}
  for all non-empty subset $S \subseteq [k]$. I.e. the \hadk{}
  predicate is satisfied if and only if the input string forms the truth table
  of a linear function.
\end{definition}

\begin{remark}
  The \threelint{} predicate and the \hadk[2]{} predicate are in fact
  identical. Thus the family of Hadamard \maxcsp{}s can be seen as a natural
  generalisation of \maxthreelint{}.
\end{remark}

\begin{remark}
  The set $\Sat(\text{\hadk{} predicate})$ can be expressed using
  a $2^k$ dimensional Hadamard matrix. Let $M_k$ be a $2^k \times (2^k - 1)$
  matrix, where the rows are index by subsets $\beta \subseteq [k]$ and the
  columns are indexed by non-empty subsets $S \subseteq [k]$. Let
  \begin{eqnarray*}
    (M_k)_{\beta, S} & = & \left\{\begin{array}{lll}
      0 & \mathrm{if} & \chi_{\beta} (S) = 1,\\
      1 & \mathrm{if} & \chi_{\beta} (S) = - 1.
    \end{array}\right.
  \end{eqnarray*}
  The matrix $M_k$ is the set $\Sat(\text{\hadk{} predicate})$
  expressed on the form of a matrix, with one row per element. Note that $M_k$
  is almost the $2^k$-dimensional Hadamard matrix. $M_k$ can be made into the
  Hadamard matrix by prepending an all 0 column to it, and then switching out
  $0 / 1$ to $1 / - 1$. This connection between $\Sat(\text{\hadk{} predicate})$ and Hadamard matrices is one of the reasons as to why
  this \maxcsp{} is called the Hadamard \maxcsp{}. An example of the matrix $M_k$
  can be found in Figure \ref{fig:mat}.
  
  \begin{figure}
  \centering
    $\left(\begin{array}{ccccccc}
      0 & 0 & 0 & 0 & 0 & 0 & 0\\
      1 & 0 & 1 & 0 & 1 & 0 & 1\\
      0 & 1 & 1 & 0 & 0 & 1 & 1\\
      1 & 1 & 0 & 0 & 1 & 1 & 0\\
      0 & 0 & 0 & 1 & 1 & 1 & 1\\
      1 & 0 & 1 & 1 & 0 & 1 & 0\\
      0 & 1 & 1 & 1 & 1 & 0 & 0\\
      1 & 1 & 0 & 1 & 0 & 0 & 1
    \end{array}\right)$
    \caption{The matrix $M_k$ for $k = 3$. It is an $8 \times 7$ matrix. Note
    that prepending a zero column to $M_k$ and switching $0/1$ to $1/-1$ would make it into a Hadamard
    matrix, which is symmetric. \label{fig:mat}}
  \end{figure}
\end{remark}

The Hadamard predicate has been shown to be a \emph{useless predicate}.
The concept of useless predicates was first introduced in {\cite{useless}}.
This property of \hadk{} was originally proven by Austrin and Mossel
using UGC {\cite{pair}}, which relies on the fact that $\Sat
(\hadk{})$ admits a balanced pairwise independent set. Later on Chan was
able to show that \hadk{} is a useless predicate without requiring UGC
{\cite{Chan}}. To state this result we first need two definitions.

\begin{definition}
  Given an instance $\mathcal{I}$ of an $m$-ary \maxcsp{} and an assignment $A$,
  let $\mathcal{D} (A, \mathcal{I})$ denote the distribution of binary strings
  $\bool^m$ generated by sampling \\
  $((x_1, b_1), \ldots, (x_m, b_m))
  \sim \mathcal{I}$ and outputting the binary string $((A (x_1) + b_1),
  \ldots, (A (x_m) + b_m))$.
\end{definition}

\begin{definition}
  The total variation distance $d_{\mathrm{TV}}$ between two probability
  measures $\mu_1$ and $\mu_2$ over a finite set $\Omega$ is defined as
  \begin{eqnarray*}
    d_{\mathrm{TV}} (\mu_1, \mu_2) & = & \frac{1}{2} \sum_{\omega \in \Omega} |
    \mu_1 (\omega) - \mu_2 (\omega) | .
  \end{eqnarray*}
\end{definition}

\begin{theorem}
  \label{thm:useless}{\cite{Chan}} For every $\varepsilon > 0$, it is NP-hard
  to distinguish between instances $\mathcal{I}$ of the
  \hadk{} \maxcsp{} such that
  \begin{description}
  \item[Completeness] There exists an assignment $A$ such that 
  \[
    d_{\mathrm{TV}} (\mathcal{D} (A, \mathcal{I}), \uniform (\{\Sat\}
    (\text{\hadk{} predicate}))) \leqslant \varepsilon.
  \]
  
  \item[Soundness] For every assignment $A$, 
  \[
    d_{\mathrm{TV}} (\mathcal{D} (A, \mathcal{I}), \uniform
    (\bool^{2^k - 1})) \leqslant \varepsilon.
  \]
  \end{description}
  Here $\uniform (\Sat (\text{\hadk{} predicate}))$ denotes
  the uniform distribution over binary strings that satisfy the \hadk{} predicate. Similarly, $\uniform (\bool^{2^k - 1})$
  denotes the uniform distribution over all binary strings of length $2^k -
  1$.
\end{theorem}

\begin{remark}
  The uniform distribution on satisfiable instances in the completeness case
  is a subtle detail. The result by {\cite{Chan}} is not formulated like this.
  However, it is trivial to take the instances constructed by Chan and modify
  them to make the completeness case be uniformly distributed over satisfied
  instances. The first time this was used was by {\cite{Wiman}}. However, this
  detail turns out to not actually matter in the end since all of the gadgets
  that we construct and all of the gadgets that Wiman construct are symmetric.
  So this uniform randomness assumption is only there because of convenience,
  and is not actually used in the end.
\end{remark}

\subsection{The automated gadget framework}\label{sec:auto}

The ``automated gadget'' framework by Trevisan et al {\cite{Trev}} describes
how to construct optimal gadgets when reducing from one predicate to another.
Let us denote the starting predicate as $\phi$ and the target predicate as
$\psi$. A \emph{$\phi$-to-$\psi$-gadget} is a description for how to
reduce a $\phi$-constraint to one or more $\psi$-constraints. As an example,
let us take a gadget from \threeSat{} to \tlint{}. In this case the gadget
describes a system of linear equations that both involve the three original
variables from the \threeSat{} constraint (called \emph{primary variables},
denoted by $\mathbb{X}$) as well as new extra variables (called
\emph{auxiliary variables}, denoted by $\mathbb{Y}$).

Gadgets have two important properties, called \emph{completeness} and
\emph{soundness}. These properties describe how closely the
$\psi$-constraints are able to mimic the satisfiability of the original
$\phi$-constraint. The completeness of a gadget is a value between $0$ and $1$
that describe how many of the $\psi$-constraints that can be satisfied under
the restriction that $\mathbb{X}$ satisfies the original $\phi$-constraint. In
a similar fashion, the soundness of a gadget is a value between $0$ and $1$
that describes the case when $\mathbb{X}$ does not satisfy the original
$\phi$-constraint. When we construct our gadgets, we fix the completeness of
the gadget, and then we find the gadget that minimises the soundness for this
fixed completeness. A gadget that minimises the soundness for a given
completeness is referred to as an \emph{optimal gadget}.

There is no a priori upper bound on how many auxiliary variables that a
$\phi$-to-$\psi$-gadget can have. However, the ``automated gadget'' framework
by Trevisan et al {\cite{Trev}} proves that, under some reasonable
assumptions, the number of variables $| \mathbb{X} \cup \mathbb{Y} |$ in an
optimal gadget can be assumed to be at most $2^{| \Sat (\phi) |}$.
Furthermore, if $\psi$ allows the negations of variables, then this number
drops to $2^{| \Sat (\phi) | - 1}$.

In the case of a \hadkttlint{} gadget, the number of
satisfying assignments of \hadk{} is $2^k$, and \tlint{}
allow the negation of variables. This means that the total number of variables
in the gadget is $2^{2^k - 1}$. Out of these, $2^k - 1$ variables are in
$\mathbb{X}$, and $2^{2^k - 1} - \left( {2^k}^{} - 1 \right)$ variables are in
$\mathbb{Y}$. Furthermore, the ``automated gadget'' framework gives a natural
way to index these variables in terms of $| \Sat (\hadk{}) |$-long
bitstrings. According to the framework, each primary variable should be
indexed by a bitstring describing that variable's assignment to all of the
satisfying assignments to \hadk{}, meaning a column in the matrix shown
in Figure \ref{fig:mat}. The auxiliary variables are indexed by the bitstrings
that do not appear in the matrix.

Instead of using $2^k$-long bitstrings to index the variables, it is arguably
more natural to index the variables using functions in $\boolk{}$. These
representations are equivalent since every $2^k$ long bitstring can be
interpreted as a truth table of a function in $\Fk_k$, and vice versa.
By indexing the set of variables using functions in $\Fk_k$, the set
of primary variables are indexed by linear functions $\{ \chi_{\alpha}
\}_{\varnothing \subset \alpha \subseteq [k]}$, and the negations of linear
functions $\{ - \chi_{\alpha} \}_{\varnothing \subset \alpha \subseteq [k]}$.
This gives us the following description of a \hadkttlint{} gadget.

\begin{definition}
  A \hadkttlint{} gadget is given by a tuple $(G,
  \mathbb{X}_k, \mathbb{Y}_k)$, where $G$ is a probability distribution over
  $\binom{\Fk_k}{2}$ where $G (f_1, f_2) = 0$ if $f_1 = - f_2$.
  $\mathbb{X}_k$ is the set of primary variables and $\mathbb{Y}_k$ is the set
  of auxiliary variables. The set of variables $\mathbb{X}_k \cup
  \mathbb{Y}_k$ are indexed by functions in $\Fk_k$, meaning
  $\mathbb{X}_k \cup \mathbb{Y}_k = \{ x_f : f \in \Fk_k \}$. A
  variable $x_f$ is a primary variable if and only if $f$ is a linear function
  or $- f$ is a linear function.
  
  The reduction from a \hadk{} constraint $\hadk{} \left( {b_{\{ 1
  \}}}  + y_{\{ 1 \}}, \ldots, b_{[k]} + y_{[k]} \right)$ to \tlint{} is given by the distribution formed by
  \begin{enumerate}
    \item Sampling $(f_1, f_2) \sim G$,
    
    \item Outputting the constraint $T (f_1) = T (f_2)$ where
  \end{enumerate}
  \begin{eqnarray*}
    T (f) & = & \left\{\begin{array}{ll}
      x_f & \text{if } x_f \in \mathbb{Y}_k,\\
      b_{\alpha}  + y_{\alpha} & \text{if } f = \chi_{\alpha} \text{ for some } \alpha \in \boolk, \\
      b_{\alpha}  + y_{\alpha} + 1 & \text{if } f = -\chi_{\alpha} \text{ for some } \alpha \in \boolk .
    \end{array}\right.
  \end{eqnarray*}
\end{definition}

Let us now precisely define the soundness and completeness of a
\hadkttlint{} gadget $(G, \mathbb{X}_k, \mathbb{Y}_k)$.
From Theorem \ref{thm:useless} it follows that the natural definition of
soundness is to uniformly at random assign the primary variables
$\mathbb{X}_k$ to $\bool$, and then assign the rest of the variables in order
to satisfy as many of the equations as possible.

\begin{definition}
  Given a set of Boolean variables $\mathbb{X}$. Let $\Fk
  (\mathbb{X})$ denote the set of assignments $\mathbb{X} \rightarrow \bool$.
  Let $\Ffold (\mathbb{X})$ the set of all folded
  assignments, i.e. functions $P : \mathbb{X} \rightarrow \bool$ such that $P
  (1 + x) = 1 + P (x) \forall x \in \mathbb{X}.$ Here $1 + x$ denotes the negation of the variable $x$.
\end{definition}

\begin{definition}
  The soundness of $G$ is defined as
  \begin{eqnarray*}
    s (G) & = & \underset{\begin{array}{l}
      P \in \Ffold (\mathbb{X}_k)
    \end{array}}{\expect} \max_{\begin{array}{l}
      A \in \Ffold (\mathbb{X}_k \cup \mathbb{Y}_k),\\
      A (x) = P (x), x \in \mathbb{X}_k
    \end{array}} \val (A, G),
  \end{eqnarray*}
  where
  \begin{eqnarray*}
    \val (A, G) & = & \sum_{\begin{array}{l}
      (f_1, f_2) \in \binom{\Fk_k}{2}
    \end{array}} G (f_1, f_2)  [A (x_{f_1}) = A (x_{f_2})] .
  \end{eqnarray*}
\end{definition}

The completeness of $G$ is defined using dictator cuts. A \emph{dictator
cut} $\delta_y$ of $y \in \boolk$ is an assignment where $(- 1)^{\delta_y
(x_f)} = f (y)$. From Theorem \ref{thm:useless} we see that that the natural
definition for completeness is the expectation over $\val (\delta_y,
G)$, where $\delta_y$ is a random dictator cut.

\begin{definition}
  \label{def:comp}The completeness of $G$ is defined as
  \begin{eqnarray*}
    \begin{array}{lllll}
      c (G) & = & \underset{\begin{array}{l}
        y \in \boolk
      \end{array}}{\expect} \val (\delta_y, G) & = & 1 -
      \sum\limits_{\begin{array}{l}
        (f_1, f_2) \in \binom{\Fk_k}{2}
      \end{array}} G (f_1, f_2) \dist (f_1, f_2) .
    \end{array} &  & 
  \end{eqnarray*}
\end{definition}

There is a result based on Theorem \ref{thm:useless} that relates the
soundness and completeness of \hadkttlint{} gadgets to
\NPhardness{} results for \maxtlint{}.

\begin{proposition}
  \label{prop:gap}{\cite[Proposition 2.17]{Has2}} Given a \hadkttlint{} gadget $(G, \mathbb{X}_k, \mathbb{Y}_k)$ with $s = s (G)$
  and $c = c (G)$, where $c > s$. Then for every sufficiently small
  $\varepsilon > 0$, it is \NPhard{} to distinguish between instances
  $\mathcal{I}$ of \maxtlint{} such that
  
  (Completeness) There exists an assignment that satisfies a fraction at
  least $c - \varepsilon$ of the constraints.
  
  (Soundness) All assignments satisfy at most a fraction $s + \varepsilon$ of
  the constraints.
\end{proposition}

One particularly interesting case is the inapproximability of \mintlint{}. From UGC it follows that it is \NPhard{} to approximate
\mintlint{} within any constant {\cite{OW}}. The
following proposition from Håstad et al. {\cite{Has2}} tells us that a
\hadkttlint{} gadget reduction can never be used to
show an inapproximability factor of \mintlint{} better than $2.54$.
This means that any \NPhardness{} result for \mintlint{} shown using
a gadget reduction from \hadkttlint{} cannot match
results obtained by UGC.

\begin{proposition}
  \label{prop:3.7}{\cite[Proposition 2.29 and Theorem 6.1]{Has2}} For any
  given \hadkttlint{} gadget $(G, \mathbb{X}_k,
  \mathbb{Y}_k)$. There exists a \hadkttlint{} gadget
  $(\tilde{G}, \mathbb{X}_k, \mathbb{Y}_k)$ with completeness $1 - 2^{- k}$
  such that
  \begin{eqnarray*}
    \frac{1 - s (G)}{1 - c (G)} & \leqslant & \frac{1 - s (\tilde{G})}{1 - c
    (\tilde{G})},
  \end{eqnarray*}
  and
  \begin{eqnarray*}
    \frac{1 - s (\tilde{G})}{1 - c (\tilde{G})} & \leqslant & \frac{1}{1 -
    e^{- 0.5}} \approx 2.54.
  \end{eqnarray*}
\end{proposition}

\begin{remark}
  A \hadkttlint{} gadget having completeness $1 - 2^{-
  k}$ implies that the gadget only have positive weight edges of length $2^-k$. So
  far fewer edges are used compared to the total number of possible edges.
\end{remark}

\begin{remark}
  The upper limit of $2.54$ shown by {\cite{Has2}} is much more general than
  what is stated here. In fact, they show that the bound of $2.54$ holds for
  any gadget reduction from a useless predicate $\phi$ such that $\Sat
  (\phi)$ has a balanced pairwise independent subset.
\end{remark}

\section{Relaxed soundness and infinity relaxed soundness  }\label{sec_gadget}

The main difficulty when designing and analysing gadgets is that the soundness
is difficult to compute. In the case of a gadget reduction from \maxhadk{}
to \maxtlint{}, computing the soundness of the gadget involves solving an
instance of \maxtlint{}. For $k \leqslant 3$ this is computationally
feasible, since the \maxtlint{} instance is rather small, but for $k \geqslant
4$ the instances can become so large that, even using a computer, it is practically impossible to solve them.

To get around this issue, Wiman {\cite{Wiman}} proposed to relax the
definition of the soundness by not requiring that the assignment $A$ of the
auxiliary variables $\mathbb{Y}_k$ is folded. Note that the assignment $A$ is
still required to be folded on the primary variables $\mathbb{X}_k$, meaning
$A (x_f) = 1 + A (x_{- f}) \forall x_f \in \mathbb{X}_k$. Removing the
requirement that $A$ is folded over $\mathbb{Y}_k$ makes it significantly
easier to compute the soundness.

\begin{definition}
  {\cite[Definition 3.3]{Wiman}} Wiman's relaxed soundness
  \[
    \rs(G) = \underset{P \in \Ffold
    (\mathbb{X}_k \cup \{ x_1, x_{- 1} \})}{\expect}
    \max_{\begin{array}{l}
      A \in \Fk (\mathbb{X}_k \cup \mathbb{Y}_k),\\
      A (x) = P (x), x \in \mathbb{X}_k \cup \{ x_1, x_{- 1} \}
    \end{array}} \val(A, G),
  \]
\end{definition}

where
\[
  \val(A, G) = \sum_{\begin{array}{l}
    (f_1, f_2) \in \binom{\Fk_k}{2}
  \end{array}} G (f_1, f_2)  [A (x_{f_1}) = A (x_{f_2})] .
\]
This relaxation fundamentally changes the soundness computation from being a
\maxtlint{} problem to being an $s$-$t$ \mincut{} problem. This is because the
computation of $1 - \rs(G)$ for a fixed $P$ is a minimisation problem
where the goal is to minimise the number of times that $A (x_{f_1}) \neq A
(x_{f_2})$, which makes it a $s$-$t$ \mincut{} problem. According to the
\maxflow{} \mincut{} Theorem, this also means that $\rs(G)$ can be computed by solving a \maxflow{} problem.

The conclusion from this is that $1 - \rs(G)$ can be interpreted as
the average max flow on the fully connected $2^k$-dimensional hypercube, where
the placement of sources and sinks is randomly distributed over nodes labeled
by affine functions. The sources correspond to primary variables $x_f$ where
$P (x_f) = 1$ and the sink nodes correspond to primary variables $x_f$ where
$P (x_f) = 0$. The capacity of an edge $\{ x_{f_1}, x_{f_2} \}$ in the fully
connected hypercube is given by $G (f_1, f_2)$. Note that the sum over
capacities in the graph is equals to 1.

There are some significant benefits to using the relaxed soundness. Firstly,
it is significantly simpler to solve a \maxflow{} problem compared to a
\maxtlint{} problem. The implication from this is that it is computationally
simple to compute the relaxed soundness of \hadkttlint[4]{} gadgets and even possible to compute relaxed soundness of
\hadkttlint[5]{} gadgets if given enough computational
resources. Furthermore, the relaxed soundness allows us to analyse
\hadkttlint{} gadgets even in the case where $k$ is
very large. The disadvantage to using relaxed soundness is that it is not guaranteed to be close to the true soundness.

\subsection{Relaxed soundness described as an LP}

Recall that $1 - \rs(G)$ can be expressed as the average max flow on a
fully connected $2^k$-dimensional hypercube with randomised source/sink
placements. This means that $\rs(G)$ can be stated as an LP. One reason
for why it is preferable to express this \maxflow{} problem as an LP is because
it is possible to move the capacities (i.e. the ``gadget variables'') of the
\maxflow{} problem to the variable side of the LP. So the same LP can be used both to calculate the the relaxed soundness of a specific gadget and to construct new gadgets.

One additional step we use in the formulation of this LP is to use a
function $g \in \Fk_k$ to describe
the source/sink placement instead of using the assignment $P$. A node $v_{\chi_{\alpha}}, \alpha \in \boolk{}$ is a
sink node if $g (\alpha) = 1$, and a source node if $g (\alpha) = - 1$. These two
representations of the source/sink placement are equivalent, but using a
Boolean function $g$ is more helpful for understanding the symmetries of the LP, as done in Appendix \ref{sec4}. The following is the LP
reformulation of the relaxed soundness $\rs (G)$.

\begin{definition}
  A flow $w$ of a \hadkttlint{} gadget $(G,\mathbb{X}_k, \mathbb{Y}_k)$ is a function $\Fk_k^3 \rightarrow
  \real_{\geqslant 0}$. The flow $w$ is said to be feasible if and only
  if
  \begin{alignat}{2}
    w (f_1, f_2, g) + w (f_2, f_1, g) & \leqslant & \, G (f_1, f_2) \quad & \forall
    f_1, f_2, g \in \Fk_k,  \label{eq1}\\
    \fout_w (f, g) & = & \, \fin_w (f, g) \quad & \forall f, g \in
    \Fk_k, \dimension (f) \geqslant 1.  \label{eq2}
  \end{alignat}
  where
  $
    \fout_w (f, g) = \sum_{f_2 \in \Fk_k} w (f, f_2, g) \text{ and }
    \fin_w (f, g) = \sum_{f_2 \in \Fk_k} w (f_2, f, g)
  $.
  The value of $w$ for at a source/sink placement $g \in \Fk_k$ is
  defined as
  \[
    \val_g (w) = \sum_{\alpha \in \boolk{}} \fout_w (g (\alpha)
    \chi_{\alpha}, g) - \fin_w (g (\alpha) \chi_{\alpha}, g) .
  \]
\end{definition}

\begin{definition}
  \label{sound}The relaxed soundness LP for a \hadkttlint{} gadget $(G, \mathbb{X}_k, \mathbb{Y}_k)$, denoted by $\rsLP (G)$, is
  the following LP
  \[
    \rs(G) = 1 - \max_w \expect_{g \in \Fk_k}
    \val_g (w),
  \]
  where the maximum is taken over feasible flows $w$ of $G$.
\end{definition}

\begin{remark}
  Recall that $1 - \rs (G)$ is the average of $2^{2^k}$ independent
  \maxflow{} problems. The different \maxflow{} problems are indexed by the
  function $g \in \Fk_k$, which describes the placements of sinks and
  sources. The nodes in each \maxflow{} problem are indexed by functions in
  $\Fk_k$. The sink nodes in the $g$-th \maxflow{} problem are the nodes
  $v_{g (\alpha) \chi_{\alpha}}, \alpha \in \boolk{}$, and the source nodes are
  the nodes $v_{- g (\alpha) \cdot \chi_{\alpha}}, \alpha \in \boolk{}$. The
  flow from $v_{f_1} \rightarrow v_{f_2}$ is $w (f_1, f_2, g)$, and the
  capacity of the undirected edge $\{ v_{f_1}, v_{f_2} \}$ is $G (f_1, f_2)$.
\end{remark}

\begin{remark}
  Note that it is possible to modify the $\rsLP (G)$ to include the capacities
  of the graph (i.e. gadget $G$) as variables. The implications of this is
  that the optimisation problem of finding a \hadkttlint{} gadget with the maximum relaxed soundness for a fixed completeness can
  also be expressed as an LP.
\end{remark}

\begin{remark}
  Note that the $\rsLP (G)$ has roughly $| \Fk_k |^3 = 2^{3 \cdot
  2^k}$ variables. This is a very large number, even for small values of $k$.
  So in order to be able to solve this LP, we have to make use of the
  symmetries of the LP in order to reduce the number of variables.
\end{remark}

\subsection{Introduction of infinity relaxed soundness}

One natural question is, how small can one make the relaxed soundness if we
fix the completeness of a \hadkttlint{} gadget and let
$k \rightarrow \infty$? In practice, even just finding the gadget minimising
the relaxed soundness when $k = 5$ is a very daunting task, so cannot hope to
calculate this limit directly from the $\rsLP (G)$.

Our method to handle large values of $k$ is to create a \hadkttlint{} gadget for some small value of $k$, for example $k = 4$, and
then introduce the concept of embedding a \hadkttlint{} gadget $G$ into a \hadkttlint[k']{} gadget $G'$,
where $k' > k$. It is also possible to embed the flow of the $\rsLP (G)$ onto
the $\rsLP(G')$. This embedding has the property that the completeness and the soundness of both gadgets are the same.

The key insight is that by using multiple overlapping embeddings of $G$, we
can improve the soundness of $G'$ without affecting its completeness. Our
argument for why multiple overlapping embeddings improve the relaxed soundness
is based on leaky flows. Note that the leaks of a leaky flow have signs. This
means that it is possible that overlapping embeddings of leaky flows could
become a feasible flow, since the overlap of the embeddings could cause the
signed leaks to sum to 0. We use this type of argument to show an upper bound
on $\rs (G')$ based on a leaky flow solution to the $\rsLP(G)$.

The exact procedure for the embeddings is defined in Appendix \ref{sec4} and
analysed in detail in Appendix \ref{sec_proof} using second moment analysis.
The conclusion from that analysis is that the following relaxation of the
$\rsLP(G)$, which we call the \emph{infinity relaxed soundness LP},
denoted by $\rsinfLP(G)$, has the following two important properties. Firstly, the
solution of $\rsinfLP(G)$ is a leaky flow of $\rsLP(G)$, and secondly,
overlapping embeddings of this leaky flow tends to a feasible flow of
$\rsLP(G')$ as $k' \rightarrow \infty$.

\begin{definition}
  \label{def:relaxed_constraint}A flow $\tilde{w}$ of a \hadkttlint{} gadget $(G, \mathbb{X}_k, \mathbb{Y}_k)$ is said to be a
  infinity relaxed flow if constraint \eqref{eq1} is satisfied and
  \begin{eqnarray}
    \sum_{g'} \fout_w (f, g') = \sum_{g'} \fin_w (f, g') \quad
    \forall g, f \in \Fk_k : \dimension (f) \geqslant 1,  \label{eq:weird}
  \end{eqnarray}
  where the sums are over functions $g' \in \Fk_k'$ such that $g' \barsuchthat_{\affine(f)} = g \barsuchthat_{\affine (f)}$. The
  (signed) leak at $(f, g)$, where $f, g \in \Fk_k, \dimension (f) \geqslant
  1$, is defined as $\leak_{\tilde{w}} (f, g) = \fin_w (f, g) -
  \fout_w (f, g)$.
\end{definition}

\begin{definition}
  The infinity relaxed soundness of $G$, denoted by $\rsinf (G)$, is the
  solution to the $\rsinfLP(G)$
  \begin{eqnarray*}
    \rsinf (G) = 1 - \max_{\tilde{w}} \expect_{g \in
    \Fk_k } \val_g (\tilde{w}),
  \end{eqnarray*}
  where the maximum is taken over all infinity relaxed flows $\tilde{w}$ of
  $G$.
\end{definition}

\begin{remark}
  The $\rsinfLP(G)$ is a constraint relaxation of the $\rsLP(G)$ where
  constraint \eqref{eq2} has been relaxed to constraint \eqref{eq:weird}. So a
  solution of the $\rsinfLP(G)$ is a leaky flow in the $\rsLP(G)$.
\end{remark}

\begin{remark}
  Constraint \eqref{eq:weird} is used for a proof in
  Appendix \ref{sec_proof} of Lemma \ref{lemma_key}, which is a 2nd order moment analysis of the
  overlap of leaks from random embeddings. In the proof, constraint
  \eqref{eq:weird} is used to show that if $w$ is an infinity relaxed flow
  then $\forall g, f \in \Fk_k, \dimension (f) \geqslant 1:$
    $\sum_{g'} \leak_w (f, g') = 0$,
  where the sum is over $g' \in \Fk_k$ such that
  $g' |_{\affine (f)} = g  |_{\affine (f)}$.
\end{remark}

The following lemma describes a relationship between the $\rsLP(G)$ and the
$\rsinfLP(G)$. This is the key Lemma, which is proven in Appendix
\ref{sec_proof}.

\begin{lemma}
  \label{thm_important}Let $(G, \mathbb{X}_k, \mathbb{Y}_k)$ be a
  \hadkttlint{} gadget. For any $\varepsilon > 0$ there
  exists a \hadkttlint[k']{} gadget $(G',
  \mathbb{X}_{k'}, \mathbb{Y}_{k'})$ for some $k' > k$ such that $c (G) = c
  (G')$ and $\rs (G') \leqslant \rsinf (G) + \varepsilon$.
\end{lemma}

From this Lemma, it follows that $\rsinf (G) + \varepsilon$ is the upper bound
of $\rs(G')$ for some gadget $G'$, which in turn is an upper bound of $s
(G')$. This means that the NP-hardness result of \maxtlint{} stated
in Proposition \ref{prop:gap} for $\rs (G)$ also holds for
$\rsinf (G) .$This gives us our main result.

\begin{theorem}
  \label{main}Let $(G, \mathbb{X}_k, \mathbb{Y}_k)$ be a \hadkttlint{} gadget with $s = \rsinf (G)$ and $c = c (G)$, where $c > s$.
  Then for every sufficiently small $\varepsilon > 0$, it is NP-hard to
  distinguish between instances of \maxtlint{} such that
  \begin{description}
      \item[Completeness] There exists an assignment that satisfies a fraction at
      least $c - \varepsilon$ of the constraints.
      \item[Soundness] All assignments satisfy at most a fraction $s + \varepsilon$ of
      the constraints.
  \end{description}
\end{theorem}

\section{Numerical results }\label{sec:num}

This section contains a presentation of constructed \hadkttlint{} gadgets. Recall that there are three different ways to measure
the soundness of a \hadkttlint{} gadget. There is the
true soundness of a gadget, which can be used to show \NPhardness{} results for
\maxtlint{}, see Proposition \ref{prop:gap}. Then there is the relaxed
soundness, denoted by $\rs$. This is an upper bound of the true
soundness, see Proposition \ref{prop:pi}. Finally there is the infinity
relaxed soundness, denoted by $\rsinf$, which according to our
main result, Theorem \ref{main}, also imply \NPhardness{} results for \maxtlint{}.

We compute gadgets for $k = 2, 3, 4$, optimised either for $\rs$ or
$\rsinf$. The short rundown of the process of constructing a
gadget is to first decide on the completeness of the gadget, and then call an
LP-solver to find the gadget with that completeness that either minimises
$\rs$ or $\rsinf$, depending on which measure of soundness
we want to optimise the gadget for.

\subsection{Edges used/unused in constructed gadgets}\label{list}

The capacity $G$ of a \hadkttlint{} gadget $(G,
\mathbb{X}_k, \mathbb{Y}_k)$ is a probability distribution over (undirected) edges.
Every gadget that we construct is symmetrical under the mappings of
$\Mk_{k \rightarrow k}$, so edges from the same edge orbit share the
same capacity. \cref{fig11,fig12,fig13} in Appendix \ref{sec:used} list all edge
orbits that have non-zero weight in at least one of our constructed gadgets for $k = 2, 3,
4$. Note that as discussed in Appendix \ref{sec:rest}, in the case of $k = 4$
it is possible that the gadgets we construct are sub-optimal if $c < 1 - 2^{-
k}$. This means that it is possible that the Table for $k = 4$, Table
\ref{fig13}, could look slightly different had we constructed optimal gadgets.

\subsection{Lists and plots of gadgets}

Figures \ref{fig21}, \ref{fig22} and \ref{fig23} show \hadkttlint{} gadgets with completeness on the $x$-axis, and either maximal
$\frac{1 - \rs (G)}{1 - c (G)}$ or maximal $\frac{1 - \rsinf
(G)}{1 - c (G)}$ \ on the $y$-axis. To create this plot, we construct one
gadget for each completeness value from $0.5$ to $1 - 2^{- k}$ (inclusive),
with a spacing of $2^{- 9}$. The curve is constructed using interpolation by
taking convex combinations of pairs of neighbouring gadgets.

\begin{figure}
  \centering
  {\includegraphics[width=0.8\textwidth]{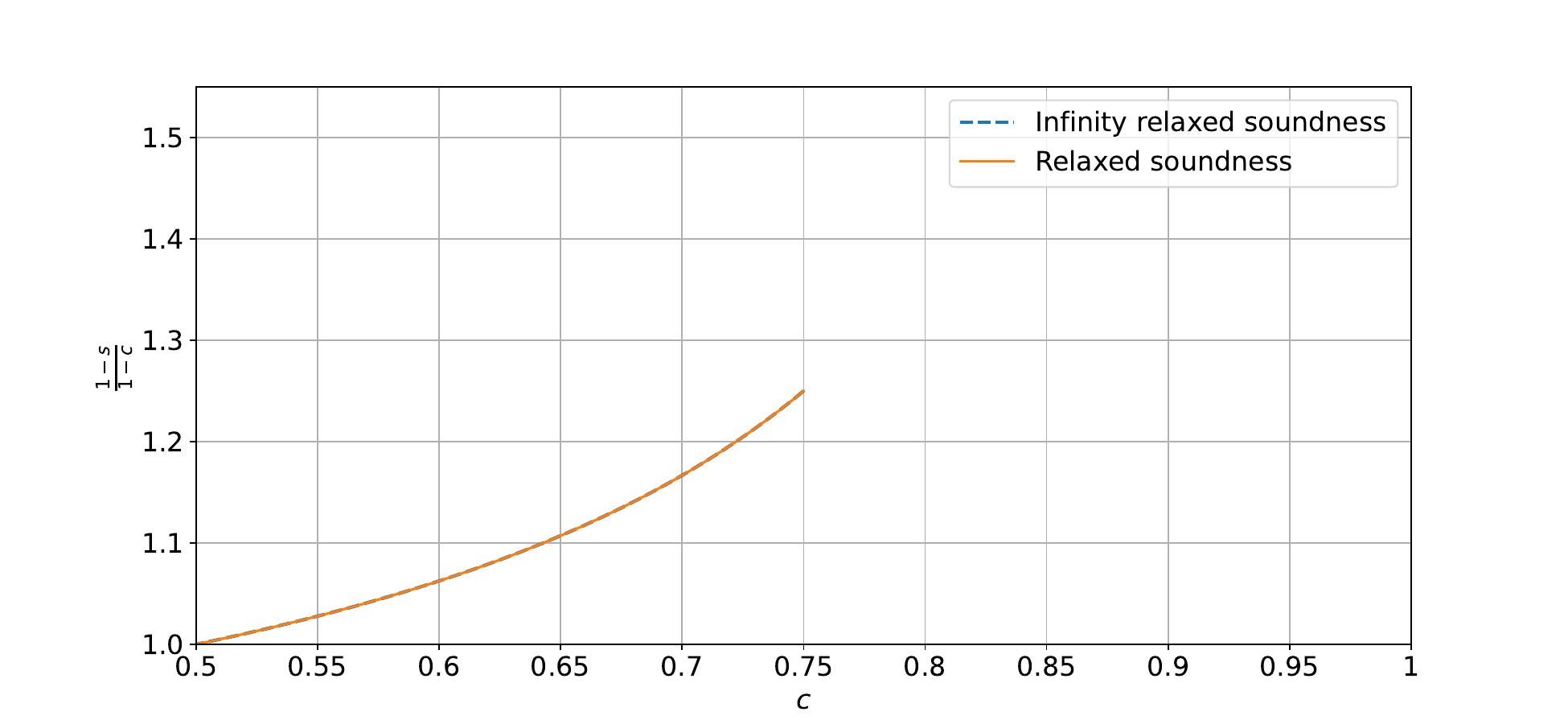}}
  \caption{This plot shows two types of \hadkttlint[2]{}
  gadgets. The filled curve describes the minimisation of $\rs$ and the
  striped curve describes the minimisation of $\rsinf$. The
  completeness value is on the $x$-axis, and either $\frac{1 - \rs (G)}{1 - c
  (G)}$ or $\frac{1 - \rsinf (G)}{1 - c (G)}$ on the $y$-axis. In this
  particular case, the case of $k = 2$, it turns out that these two curves are
  identical. \label{fig21}}
  {\includegraphics[width=0.8\textwidth]{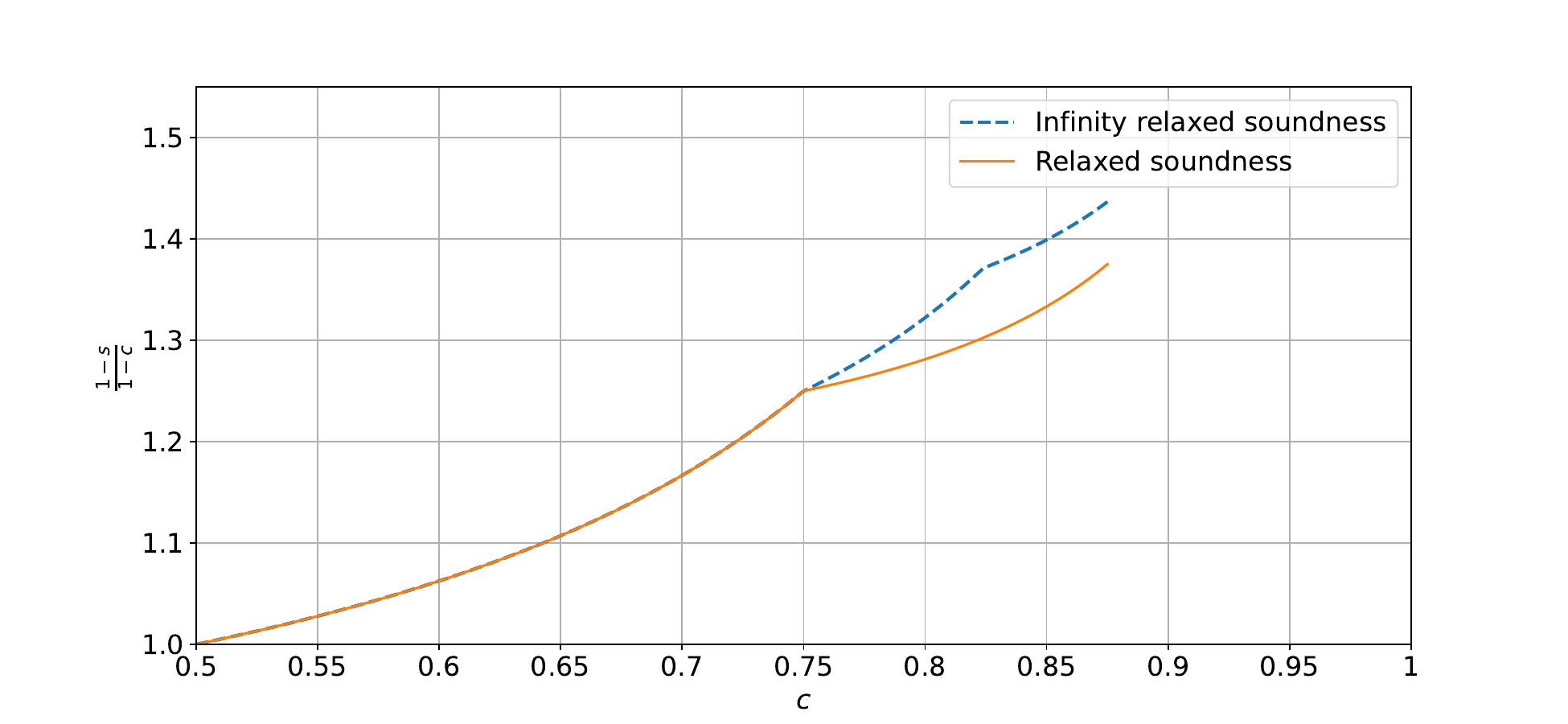}}
  \caption{This plot shows two types of \hadkttlint[3]{}
  gadgets. The filled curve describes the minimisation of $\rs$ and the
  striped curve describes the minimisation of $\rsinf$. The
  completeness value is on the $x$-axis, and either $\frac{1 - \rs
  (G)}{1 - c (G)}$ or $\frac{1 - \rsinf (G)}{1 - c (G)}$ on the
  $y$-axis. \label{fig22}}
  {\includegraphics[width=0.8\textwidth]{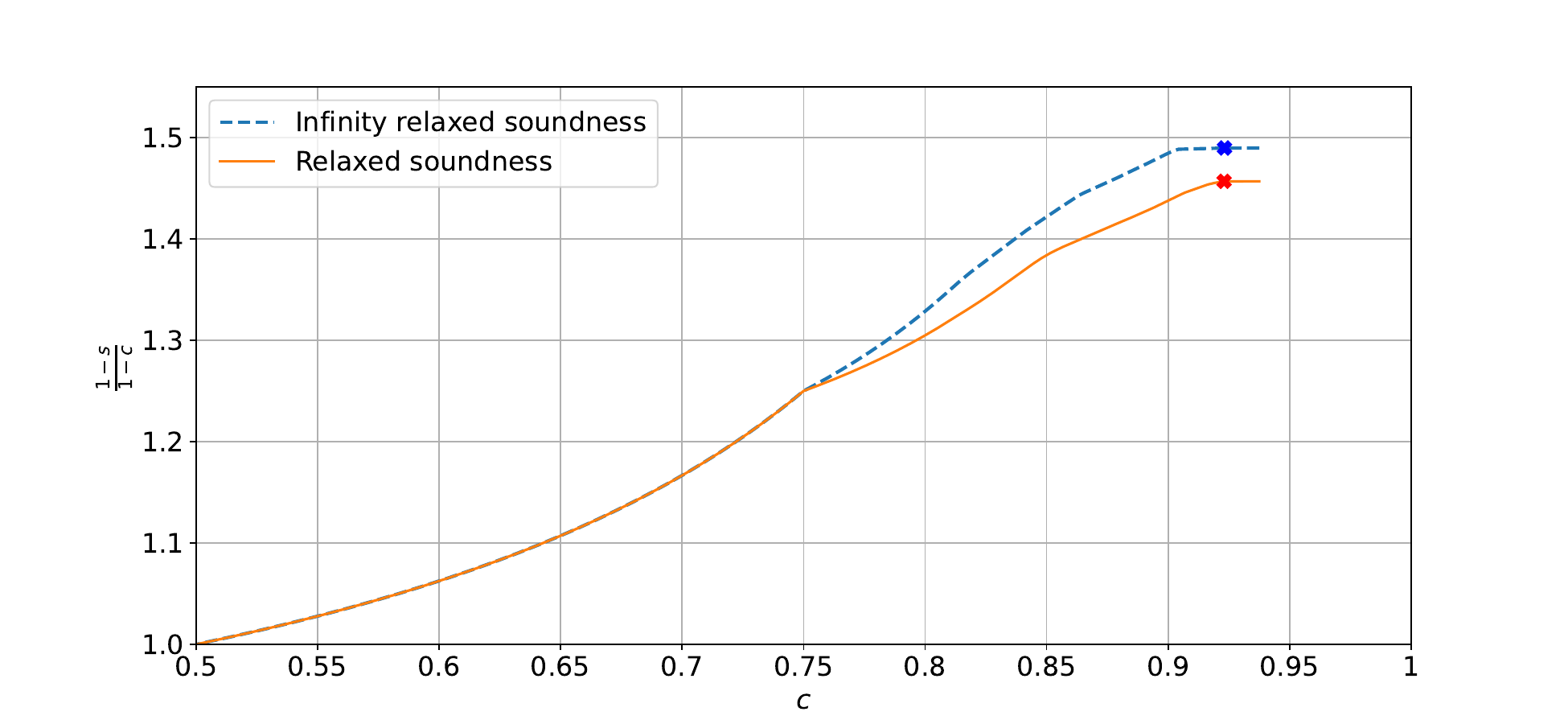}}
  \caption{This plot shows two types of \hadkttlint[4]{}
  gadgets. The filled curve describes the minimisation of $\rs$ and the
  striped curve describes the minimisation of $\rsinf$. The
  completeness value is on the $x$-axis, and either $\frac{1 - \rs
  (G)}{1 - c (G)}$ or $\frac{1 - \rsinf (G)}{1 - c (G)}$ on the
  $y$-axis. The top part of both of these curves are perfectly flat, which is
  not the case in Figure \ref{fig21} and Figure \ref{fig22}. The gadgets that
  mark the point where the curves become flat can be found in
  \cref{tab:rsound,tab:rsoundinf}, and are marked by crosses in the plot. \label{fig23}}
\end{figure}

\subsubsection{The curve \texorpdfstring{$s (c)$}{s(c)}}

The curve $s (c)$ describes the infinity relaxed soundness of
\hadkttlint[4]{} gadgets as a function of completeness,
shown as the upper curve in Figure \ref{fig23}, as well as in Figures
\ref{fig1}, \ref{fig2} and \ref{fig3}. The data for this curve can be found in
Table \ref{tab:sc}. It has the following formal definition.

\begin{definition}
  \label{curve}The curve $s (c) : [0.5, 1] \rightarrow [0.5, 1]$, $k = 4$, is
  for $c \in [0.5, 1 - 2^{- k}]$ defined as the solution to the restricted
  compressed $\rsinfLP$. For $c \geqslant 1 - 2^{- k}$ the curve is defined
  as $s (c) = 1 + 2^k (s (1 - 2^{- k}) - 1) (1 - c)$, meaning $\frac{1 - s
  (c)}{1 - c}$ is constant for all $c \geqslant 1 - 2^{- k}$.
\end{definition}

\begin{proof}[Proof of Theorem \ref{main_res}]
  For $c \in [0.5, 1 - 2^{- k}]$, the NP-hardness result follows directly from
  Theorem \ref{main} since the solution of the restricted compressed $\rsinfLP
  (G)$ is an upper bound of the (non-restricted) $\rsinfLP (G)$. For $c
  \geqslant 1 - 2^{- k}$ the \NPhardness{} result follows from taking the convex
  combination of $(c, s) = (1 - 2^{- k}, s (1 - 2^{- k}))$ and $(c, s) = (1,
  1)$. Since it is possible to create a hard instance by taking the convex
  combination of two hard instances using separate variables.
\end{proof}

\begin{table}
  \small
  \centering
  $\begin{array}{|c|c|c|c|c|c|c|c|c|c|c|c|c|c|}
    \hline
    c & s (c) &  & c & s (c) &  & c & s (c) &  & c & s (c) &  & c & s (c)\\
    0.500 & 0.5000 &  & 0.600 & 0.5750 &  & 0.700 & 0.6500 &  & 0.800 & 0.7343
    &  & 0.900 & 0.8516\\
    0.505 & 0.5038 &  & 0.605 & 0.5788 &  & 0.705 & 0.6538 &  & 0.805 & 0.7390
    &  & 0.905 & 0.8586\\
    0.510 & 0.5075 &  & 0.610 & 0.5825 &  & 0.710 & 0.6575 &  & 0.810 & 0.7437
    &  & 0.910 & 0.8661\\
    0.515 & 0.5113 &  & 0.615 & 0.5863 &  & 0.715 & 0.6613 &  & 0.815 & 0.7485
    &  & 0.915 & 0.8735\\
    0.520 & 0.5150 &  & 0.620 & 0.5900 &  & 0.720 & 0.6650 &  & 0.820 & 0.7535
    &  & 0.920 & 0.8809\\
    0.525 & 0.5188 &  & 0.625 & 0.5938 &  & 0.725 & 0.6688 &  & 0.825 & 0.7588
    &  & 0.925 & 0.8884\\
    0.530 & 0.5225 &  & 0.630 & 0.5975 &  & 0.730 & 0.6725 &  & 0.830 & 0.7642
    &  & 0.930 & 0.8958\\
    0.535 & 0.5263 &  & 0.635 & 0.6013 &  & 0.735 & 0.6763 &  & 0.835 & 0.7696
    &  & 0.935 & 0.9032\\
    0.540 & 0.5300 &  & 0.640 & 0.6050 &  & 0.740 & 0.6800 &  & 0.840 & 0.7752
    &  & 0.940 & 0.9107\\
    0.545 & 0.5338 &  & 0.645 & 0.6088 &  & 0.745 & 0.6838 &  & 0.845 & 0.7809
    &  & 0.945 & 0.9181\\
    0.550 & 0.5375 &  & 0.650 & 0.6125 &  & 0.750 & 0.6875 &  & 0.850 & 0.7868
    &  & 0.950 & 0.9256\\
    0.555 & 0.5413 &  & 0.655 & 0.6163 &  & 0.755 & 0.6922 &  & 0.855 & 0.7927
    &  & 0.955 & 0.9330\\
    0.560 & 0.5450 &  & 0.660 & 0.6200 &  & 0.760 & 0.6969 &  & 0.860 & 0.7988
    &  & 0.960 & 0.9405\\
    0.565 & 0.5488 &  & 0.665 & 0.6238 &  & 0.765 & 0.7016 &  & 0.865 & 0.8050
    &  & 0.965 & 0.9479\\
    0.570 & 0.5525 &  & 0.670 & 0.6275 &  & 0.770 & 0.7063 &  & 0.870 & 0.8115
    &  & 0.970 & 0.9554\\
    0.575 & 0.5563 &  & 0.675 & 0.6313 &  & 0.775 & 0.7109 &  & 0.875 & 0.8181
    &  & 0.975 & 0.9628\\
    0.580 & 0.5600 &  & 0.680 & 0.6350 &  & 0.780 & 0.7156 &  & 0.880 & 0.8247
    &  & 0.980 & 0.9703\\
    0.585 & 0.5638 &  & 0.685 & 0.6388 &  & 0.785 & 0.7203 &  & 0.885 & 0.8313
    &  & 0.985 & 0.9777\\
    0.590 & 0.5675 &  & 0.690 & 0.6425 &  & 0.790 & 0.7250 &  & 0.890 & 0.8380
    &  & 0.990 & 0.9852\\
    0.595 & 0.5713 &  & 0.695 & 0.6463 &  & 0.795 & 0.7297 &  & 0.895 & 0.8448
    &  & 0.995 & 0.9926\\
    \hline
  \end{array}$
  \caption{The curve $s (c)$ as shown in Figure \ref{fig1}. The values of $s
  (c)$ in this table are rounded up to $4$ decimals. This table has the same
  format as the table describing the $\mathrm{Gap}_{\mathrm{SDP}} (c)$ curve,
  found in Appendix E of {\cite{OW}}. \label{tab:sc}}
\end{table}

\subsection{Notable gadgets}

There are two gadgets that are of particular interest. These are the gadgets
with minimal completeness among those that maximises either $\frac{1 -
\rs (G)}{1 - c (G)}$ or $\frac{1 - \rsinf (G)}{1 - c (G)}$. These gadgets are marked by crosses in Figure \ref{fig23}.
The gadget with minimal completeness that maximises $\frac{1 - \rs
(G)}{1 - c (G)}$ can be found in Table \ref{tab:rsound}. The gadget with
minimal completeness that maximises $\frac{1 - \rsinf (G)}{1 - c
(G)}$ can be found in Table \ref{tab:rsoundinf}, and is also marked by a cross on the
curve $s (c)$ in Figures \ref{fig1}-\ref{fig3}. The method used to construct
such minimal completeness gadgets is slightly different compared to the
construction of gadgets with fixed completeness. Propositions \ref{prop:3.7}
and \ref{prop:pi} guarantees that gadgets with completeness $1 - 2^{- k}$ can
be used to maximise $\frac{1 - \rs (G)}{1 - c (G)}$ and $\frac{1 -
\rsinf (G)}{1 - c (G)}$. This means that the maximum values of
$\frac{1 - \rs (G)}{1 - c (G)}$ and $\frac{1 - \rsinf (G)}{1
- c (G)}$ can be computed by fixing the completeness to $c(G) = 1-2^{-k}$. Using these maximums, it is possible to slightly modify the
objective of the LP such that its solution is the gadget with minimal
completeness that maximises either $\frac{1 - \rs (G)}{1 - c (G)}$ or
$\frac{1 - \rsinf (G)}{1 - c (G)}$.

\begin{table}
  \centering
  $
  \begin{array}{lllll}
    f_1 & f_2 & \mathrm{length} & G (f_1, f_2) & \text{\% of total} \\
    \hline
    1100000000000000 & 1110000000000000 & 1 & 5461/969636864 & 30.3\\
    1110000000000000 & 1111000000000000 & 1 & 17007/1616061440 & 18.9\\
    1110000000000000 & 1110100000000000 & 1 & 437/404015360 & 23.2\\
    1110100000000000 & 1110100010000000 & 1 & 19/92346368 & 4.4\\
    0000000000000000 & 1100000000000000 & 2 & 13/215360 & 23.2\\
    \hline
  \end{array}
  $
  \caption{\label{tab:rsound}The \hadkttlint[4]{} gadget
  $G$ with minimal completeness among those that minimise $\frac{1 - \rs
  (G)}{1 - c (G)}$. The completeness of $G$ is $c (G) = 9939 / 10768$ and
  relaxed soundness is $\rs (G) = 2623643487 / 2955083776$. The right
  most column tells how much of the total capacity is contained in each edge
  orbit. This column sums up to $100 \%$. }

  $
  \begin{array}{lllll}
    f_1 & f_2 & \mathrm{length} & G (f_1, f_2) & \text{\% of total} \\
    \hline
    1100000000000000 & 1110000000000000 & 1 & 4899 / 799089790 & 33.0\\
    1110000000000000 & 1111000000000000 & 1 & 11843 / 799089790 & 26.5\\
    1110000000000000 & 1110100000000000 & 1 & 1427 / 1917815496 & 16.0\\
    1110100000000000 & 1110100010000000 & 1 & 1427 / 19178154960 & 1.60\\
    0000000000000000 & 1100000000000000 & 2 & 6094929 / 102283493120 &
    22.9\\
    \hline
  \end{array}
  $
  \caption{\label{tab:rsoundinf}The \hadkttlint[4]{}
  gadget $G$ with minimal completeness among those that minimise $\frac{1 -
  \rsinf (G)}{1 - c (G)}$. The completeness of $G$ is $c (G) =
  590174949 / 639271832$ and the infinity relaxed soundness is
  $\rsinf (G) = 141533171 / 159817958$. The right most column
  tells how much of the total capacity is contained in each edge orbit. This
  column sums up to 100\%. }
\end{table}

\section{Conclusions}

In this work, we have introduced a procedure called lifting for taking a
\hadkttlint{} gadget for a fixed $k$ and using that
gadget to construct better and better \hadkttlint[k']{}
gadgets, as $k'$ tends to infinity. In order to be able to analyse this,
both numerically and analytically, we made use of a relaxation of the (true)
soundness, first introduced by Wiman {\cite{Wiman}} in their analysis of the
\hadkttlint[4]{} gadget. This procedure allowed us to show
new inapproximability results of \maxtlint{}, and most notably using $k = 4$,
we have shown that \mintlint{} has an inapproximability factor of
$\frac{73139148}{49096883} \approx 1.48969$.

Some open problems still remain. The most obvious one is that it is likely
within reach to carry out the analysis we did for $k = 4$ also for $k = 5$.
The main bottleneck is to find or write a very efficient LP solver that is
able to handle large instances and give consistent and stable results. The
solvers available to us were not quite able to get trustworthy results. This
being said, without substantial new ideas we do not see how to attack the $k = 6$ case.

Another open problem is to understand the best possible gadget reduction from
\hadkttlint{} as $k \rightarrow \infty$. More
specifically, which is the best possible inapproximability factor of
\mintlint{} attainable using such a gadget reduction? We were able
to show an inapproximability factor of $\frac{73139148}{49096883} \approx
1.48969$ using relaxed soundness. We have also shown that by using relaxed
soundness, it is impossible to go above 2 (see Proposition \ref{prop:pi}).
Furthermore, it is known from a previous work {\cite[Theorem 6.1]{Has2}} that
by using (non-relaxed) soundness, $\frac{1}{1 - e^{- 0.5}} \approx 2.54$ is an
upper bound. This leaves us with a fairly large gap. So it would be of
interest to close this gap.

In comparison, by assuming the Unique Games Conjecture (UGC), it is possible
to show that the inapproximability factor of \mintlint{} can be
made arbitrarily large. The main open problem here is to show this without
assuming UGC. This however, is not possible to do using a gadget reduction
from \hadkttlint{}, and would instead require a
completely new approach.

Finally, as a concluding remark, it would be interesting to see if our ideas
of lifting small gadgets and analysing them using a relaxed version of the
(true) soundness, could be used in other applications. Maybe there are other
gadgets out there that could be improved using a similar procedure?

{\small
\bibliographystyle{alphaurl}
\bibliography{mybibl}

\newcommand{\etalchar}[1]{$^{#1}$}
\begin{thebibliography}{KKMO04}

\bibitem[ABS10]{subexp}
Sanjeev Arora, Boaz Barak, and David Steurer.
\newblock Subexponential algorithms for unique games and related problems.
\newblock In {\em 2010 IEEE 51st Annual Symposium on Foundations of Computer Science}, volume~62, pages 563--572, 10 2010.
\newblock \href {https://doi.org/10.1109/FOCS.2010.59} {\path{doi:10.1109/FOCS.2010.59}}.

\bibitem[ACDE]{LPSOLVER}
David Applegate, William Cook, Sanjeeb Dash, and Daniel Espinoza.
\newblock Qsopt\_ex rational lp solver.
\newblock \url{https://www.math.uwaterloo.ca/~bico/qsopt/ex/}.
\newblock Accessed: 2023-09-20.

\bibitem[AH12]{useless}
Per Austrin and Johan Håstad.
\newblock On the usefulness of predicates.
\newblock In {\em 2012 IEEE 27th Conference on Computational Complexity}, pages 53--63, 2012.
\newblock \href {https://doi.org/10.1109/CCC.2012.18} {\path{doi:10.1109/CCC.2012.18}}.

\bibitem[AM08]{pair}
Per Austrin and Elchanan Mossel.
\newblock Approximation resistant predicates from pairwise independence.
\newblock {\em CoRR}, abs/0802.2300, 2008.
\newblock URL: \url{http://arxiv.org/abs/0802.2300}, \href {https://arxiv.org/abs/0802.2300} {\path{arXiv:0802.2300}}.

\bibitem[BGS98]{BGS}
Mihir Bellare, Oded Goldreich, and Madhu Sudan.
\newblock Free bits, pcps, and nonapproximability---towards tight results.
\newblock {\em SIAM Journal on Computing}, 27(3):804--915, 1998.
\newblock \href {https://doi.org/10.1137/S0097539796302531} {\path{doi:10.1137/S0097539796302531}}.

\bibitem[Cha13]{Chan}
Siu~On Chan.
\newblock Approximation resistance from pairwise independent subgroups.
\newblock In {\em Proceedings of the Forty-Fifth Annual ACM Symposium on Theory of Computing}, STOC '13, page 447–456, New York, NY, USA, 2013. Association for Computing Machinery.
\newblock \href {https://doi.org/10.1145/2488608.2488665} {\path{doi:10.1145/2488608.2488665}}.

\bibitem[FL06]{rpr2}
Uriel Feige and Michael Langberg.
\newblock The rpr2 rounding technique for semidefinite programs.
\newblock {\em Journal of Algorithms}, 60(1):1--23, 2006.
\newblock URL: \url{https://www.sciencedirect.com/science/article/pii/S0196677404001580}, \href {https://doi.org/10.1016/j.jalgor.2004.11.003} {\path{doi:10.1016/j.jalgor.2004.11.003}}.

\bibitem[GW95]{GW}
Michel~X. Goemans and David~P. Williamson.
\newblock Improved approximation algorithms for maximum cut and satisfiability problems using semidefinite programming.
\newblock {\em J. ACM}, 42(6):1115–1145, nov 1995.
\newblock \href {https://doi.org/10.1145/227683.227684} {\path{doi:10.1145/227683.227684}}.

\bibitem[HHM{\etalchar{+}}15]{Has2}
Johan Håstad, Sangxia Huang, Rajsekar Manokaran, Ryan O’Donnell, and John Wright.
\newblock {Improved NP-Inapproximability for 2-Variable Linear Equations}.
\newblock In Naveen Garg, Klaus Jansen, Anup Rao, and Jos{\'e} D.~P. Rolim, editors, {\em Approximation, Randomization, and Combinatorial Optimization. Algorithms and Techniques (APPROX/RANDOM 2015)}, volume~40 of {\em Leibniz International Proceedings in Informatics (LIPIcs)}, pages 341--360, Dagstuhl, Germany, 2015. Schloss Dagstuhl--Leibniz-Zentrum fuer Informatik.
\newblock URL: \url{http://drops.dagstuhl.de/opus/volltexte/2015/5311}, \href {https://doi.org/10.4230/LIPIcs.APPROX-RANDOM.2015.341} {\path{doi:10.4230/LIPIcs.APPROX-RANDOM.2015.341}}.

\bibitem[Hå97]{Has}
Johan Håstad.
\newblock Some optimal inapproximability results.
\newblock In {\em Proceedings of the Twenty-Ninth Annual ACM Symposium on Theory of Computing}, STOC '97, page 1–10, New York, NY, USA, 1997. Association for Computing Machinery.
\newblock \href {https://doi.org/10.1145/258533.258536} {\path{doi:10.1145/258533.258536}}.

\bibitem[Kho02]{Khot}
Subhash Khot.
\newblock On the power of unique 2-prover 1-round games.
\newblock In {\em Proceedings of the Thiry-Fourth Annual ACM Symposium on Theory of Computing}, STOC '02, page 767–775, New York, NY, USA, 2002. Association for Computing Machinery.
\newblock \href {https://doi.org/10.1145/509907.510017} {\path{doi:10.1145/509907.510017}}.

\bibitem[KKMO04]{KKMO}
S.~Khot, G.~Kindler, E.~Mossel, and R.~O'Donnell.
\newblock Optimal inapproximability results for max-cut and other 2-variable csps?
\newblock In {\em 45th Annual IEEE Symposium on Foundations of Computer Science}, pages 146--154, 2004.
\newblock \href {https://doi.org/10.1109/FOCS.2004.49} {\path{doi:10.1109/FOCS.2004.49}}.

\bibitem[KMS18]{2to2}
Subhash Khot, Dor Minzer, and Muli Safra.
\newblock Pseudorandom sets in grassmann graph have near-perfect expansion.
\newblock In {\em 2018 IEEE 59th Annual Symposium on Foundations of Computer Science (FOCS)}, pages 592--601, 10 2018.
\newblock \href {https://doi.org/10.1109/FOCS.2018.00062} {\path{doi:10.1109/FOCS.2018.00062}}.

\bibitem[OW08]{OW}
Ryan O'Donnell and Yi~Wu.
\newblock An optimal sdp algorithm for max-cut, and equally optimal long code tests.
\newblock In {\em Proceedings of the Fortieth Annual ACM Symposium on Theory of Computing}, STOC '08, page 335–344, New York, NY, USA, 2008. Association for Computing Machinery.
\newblock \href {https://doi.org/10.1145/1374376.1374425} {\path{doi:10.1145/1374376.1374425}}.

\bibitem[TSSW96]{Trev}
L.~Trevisan, G.B. Sorkin, M.~Sudan, and D.P. Williamson.
\newblock Gadgets, approximation, and linear programming.
\newblock In {\em Proceedings of 37th Conference on Foundations of Computer Science}, pages 617--626, 1996.
\newblock \href {https://doi.org/10.1109/SFCS.1996.548521} {\path{doi:10.1109/SFCS.1996.548521}}.

\bibitem[Wim18]{Wiman}
Mårten Wiman.
\newblock Improved inapproximability of max-cut through min-cut.
\newblock Master's thesis, KTH, School of Electrical Engineering and Computer Science (EECS), 2018.

\end{thebibliography}
}

\newpage

\appendix\section{Max-Flow and symmetries\label{2.3}}

This section introduces the concepts of feasible flows and
leaky flows, and show how to make use symmetries in a graph to more
efficiently solve the \maxflow{} problem. These \maxflow{} techniques and concepts
are used during the construction of gadgets. These techniques
are very general, and become easier to explain without involving the
intricacies of gadgets. Let us start by defining the \maxflow{} problem as an LP.

\begin{definition}
  \label{def:flowgraph}A flow graph is a tuple \ $G = (V, C, S, T)$, where $C
  (u, v) = C (v, u) \geqslant 0$ is the capacity of edge $(u, v) \in V \times
  V$, and $S \subset V$ is a set of sources and $T \subset V$ is a set of
  sinks, and $S \cap T = \varnothing$.
\end{definition}

\begin{definition}
  \label{def_maxflow}A flow $w$ of a flow graph $G = (V, C, S, T)$ is a
  function $V \times V \rightarrow \real_{\geqslant 0}$. The flow $w$ is
  said to be feasible if and only if
  \begin{eqnarray}
    w (v, u) + w (u, v) & \leqslant & C (u, v) \quad \forall v, u \in V, 
    \label{eqA}\\
    \fout_w (v) & = & \fin_w (v) \quad \forall v \in V
    \setminus (S \cup T) .  \label{eqB}
  \end{eqnarray}
  where
  \begin{eqnarray*}
    \fout_w (v) & = & \sum_{u \in V} w (v, u),\\
    \fin_w (v) & = & \sum_{u \in V} w (u, v) .
  \end{eqnarray*}
  The value of a flow is defined as
  \begin{eqnarray*}
    \val (w) & = & \sum_{s \in S} \fout_w (s) - \fin_w (s) .
  \end{eqnarray*}
\end{definition}

The value of the maximum flow of a flow graph $G$ is denoted by
$\maxflowval (G)$.

\subsection{Feasible flows and leaky flows}\label{sec:leak}

When solving a \maxflow{} problem we normally require the flow to be conserved
(constraint \eqref{eqB} above), meaning that the incoming flow into a node is
equal to the outgoing flow. This is the definition of a feasible flow.
However, to find an approximate solution to a \maxflow{} problem, it can be
helpful to relax the conservation of flows constraint, allowing for ``leaks''.
A flow that does not fulfil the conservation of flow constraint is called a
\emph{leaky flow}. This section aims to analyse the relation between leaky
flows and feasible flows, with the goal of showing that if the leaks of a
leaky flow are small, then there is a feasible flow with almost the same value
as the leaky flow.

\begin{definition}
  A flow $\tilde{w}$ is said to be a leaky flow if constraint \eqref{eqA} is
  satisfied. The (signed) leak at node $v$ be defined as
  $\leak_{\tilde{w}} (v) = \fin_{\tilde{w}} (v) - \fout_{\tilde{w}} (v)$ 
  for $v \in V \setminus (S \cup T)$. 
\end{definition}

\begin{remark}
  Note that a leaky flow $\tilde{w}$ is also a feasible flow if and only if
  $\leak_{\tilde{w}} (v) = 0$ for all $v \in V \setminus (S \cup T)$.
\end{remark}

The following theorem tells us that if the sum of absolute values of the leaks
are small, then there is a feasible flow having almost the same value as the
leaky flow. The implications from this is that we can use leaky flows to get
an approximation of the true \maxflow{}.

\begin{theorem}
  \label{leaky}Given a leaky flow $\tilde{w}$ of a flow graph $G = (V, C, S, T)$,
  there exists a feasible flow $w$ of $G$ such that
  \begin{eqnarray*}
    \val (w) & \geqslant & \val (\tilde{w}) - \sum_{v \in V
    \setminus (S \cup T)} | \leak_{\tilde{w}} (v) |.
  \end{eqnarray*}
\end{theorem}

\begin{proof}
  Create a new graph $\tilde{G} = (V \cup \{ \tilde{s}, \tilde{t} \},
  \tilde{C}, S \cup \{ \tilde{s} \}, T \cup \{ \tilde{t} \})$ with an
  additional new source node $\tilde{s}$ and sink node $\tilde{t}$. We
  construct $\tilde{C}$ using $C$. Firstly let $\tilde{C} (u, v) = C (u, v)$
  for all nodes $u, v \in V$. Secondly, for every \ $v \in V \setminus (S \cup
  T)$ such that $\leak_{w'} (v) > 0$, let $\tilde{C} (u, \tilde{t}) =
  \leak_{w'} (v)$, and for every $v \in V \setminus (S \cup T)$ such
  that $\leak_{w'} (v) < 0$ let $\tilde{C} (u, \tilde{s}) = -
  \leak_{\tilde{w}} (v)$. Finally let $\tilde{C}$ be $0$ in all other
  cases.
  
  Note that for this new graph $\tilde{G}$, the leaky flow $\tilde{w}$ can be
  extended into a feasible flow since all of the leaks can be routed to either
  $\tilde{s}$ or $\tilde{t}$ depending on the sign of the leakage.
  Furthermore, if we can show that
  \begin{eqnarray}
    \maxflowval (\tilde{G}) & \leqslant & \maxflowval (G) + \sum_{v
    \in V \setminus (S \cup T)} | \leak_{\tilde{w}} (v) |
    ,  \label{eq:leak}
  \end{eqnarray}
  then that would imply the the Theorem.
  
  To show \eqref{eq:leak} we use the \maxflow{} \mincut{} Theorem. Note that any
  $S$-$T$ cut in $\tilde{G}$ has a corresponding $S$-$T$ cut in $G$ and vice
  versa since $G$ and $\tilde{G}$ share the same non-source/sink nodes.
  Additionally, note that the value of a $S$-$T$ cut in $\tilde{G}$ can be
  bounded from above by the value of the corresponding cut in $G$ plus the
  extra capacities in $\tilde{G}$. The conclusion from this is that
  \begin{eqnarray*}
    \maxflowval (\tilde{G}) & = & \mincutval (\tilde{G})\\
    & \leqslant & \mincutval (G) + \sum_{v \in V \setminus (S \cup T)} |
    \leak_{\tilde{w}} (v) |\\
    & = & \maxflowval (G) + \sum_{v \in V \setminus (S \cup T)} |
    \leak_{\tilde{w}} (v) | .
  \end{eqnarray*}
  
\end{proof}

\subsection{Symmetries of \maxflow{} graphs}\label{sec:flowsym}

If a flow graph $G = (V, C, S, T)$ has some kind of symmetry, then we can use
them to more efficiently solve the \maxflow{} problem. In our setting, the
symmetries are described by a group $H$ acting on $V$ with the property that
the capacities are invariant under the group action, meaning $C (u, v) = C (h
\cdot u, h \cdot v)$ for all $h \in H$ and $u, v \in V$. Here $h \cdot u$
denotes the group action of $h$ on $u$.

\begin{definition}
  \label{def_symgroup}Given a flow graph $G = (V, C, S, T)$ and a group $H$
  acting on $V$, then $H$ is said to be a symmetry group of $G$ if and only if
  $\forall h \in H$:
  \begin{enumerate}
    \item $h \cdot s \in S \, \forall s \in S$,
    
    \item $h \cdot t \in T \, \forall t \in T$,
    
    \item $C (u, v) = C (h \cdot u, h \cdot v) \forall h \in H$ and $\forall
    u, v \in V$.
  \end{enumerate}
\end{definition}

Using $G$ and the group $H$ acting on $V$, we can create a new flow graph
where the set of vertices is the quotient space $V / H$. This ``compresses''
the graph $G$ into one vertex per orbit. Let the capacities between two orbits
$A, B \in V / H$ be the sum capacities over all pairs in $A \times B$.

\begin{definition}
  Given a flow graph $G = (V, C, S, T) $and a symmetry group $H$ of $G$. Let
  the quotient flow graph $G / H = (V / H, C / H, S / H, T / H)$ where $V / H$
  is the set of all orbits of $V$ under the action of $H$, and similarly $S /
  H$ is the set of orbits of $S$ and $T / H$ is the set of orbits of $T$. Let
  $C / H$ be defined as a function $V / H \times V / H \rightarrow \real$
  such that
  \begin{eqnarray*}
    (C / H) (A, B) & = & \sum_{u \in A} \sum_{v \in B} C (u, v)
  \end{eqnarray*}
  for all $A, B \in V / H$.
\end{definition}

What remains to show is that the original graph $G$ and the compressed graph
$G / H$ has the same \maxflow{}.

\begin{theorem}
  \label{thm_symmax}Given a flow graph $G = (V, C, S, T) $and a symmetry group
  $H$ of $G$. Then $\maxflowval (G) = \maxflowval (G / H)$.
\end{theorem}

\begin{proof}
  First let us show that $\maxflowval (G) \leqslant \maxflowval (G /
  H)$. Let $w$ be the max-flow of $G$. Now define $w / H$ as a function from
  $V / H \times V / H \rightarrow \real$ such that
  \begin{eqnarray*}
    (w / H) (A, B) & = & \sum_{a \in A} \sum_{b \in B} w (a, b) .
  \end{eqnarray*}
  What remains to show is that that $w / H$ is a feasible flow of $G / H$ and
  that $\val (w) = \val (w / G)$ since those two properties would
  imply that $\maxflowval (G) \leqslant \maxflowval (G / H)$.
  Firstly, note that $w / H$ fulfills \eqref{eqA} and \eqref{eqB} from
  Definition \ref{def_maxflow} for the graph $G / H$ since the constraints are
  linear. For example take constraint \eqref{eqA},
  \begin{eqnarray*}
    (w / H) (A, B) + (w / H) (B, A) & = & \sum_{a \in A} \sum_{b \in B} w (a,
    b) + w (b, a)\\
    & \leqslant & \sum_{a \in A} \sum_{b \in B} C (a, b)\\
    & = & (C / H) (A, B) .
  \end{eqnarray*}
  So $w / H$ is a feasible flow of $G / H$. Secondly note that the value of
  $w$ is the same as the value of $w / H$ since
  \begin{eqnarray*}
    \val (w / H) & = & \sum_{A \in S / H} \fout_{w / H} (A) -
    \fin_{w / H} (A)\\
    & = & \sum_{A \in S / H} \sum_{s \in A} \fout_w (s) - \fin_w
    (s)\\
    & = & \sum_{s \in S} \fout_w (s) - \fin_w (s)\\
    & = & \val (w) .
  \end{eqnarray*}
  It remains to show that $\maxflowval (G) \geqslant \maxflowval (G
  / H)$. Let $w'$ be a max-flow of $G / H$. Now define $w : V \times V
  \rightarrow \real$ such that
  \begin{eqnarray*}
    w (a, b) & = & w' (H \cdot a, H \cdot b) \frac{C (a, b)}{(C / H) (H \cdot
    a, H \cdot b)}
  \end{eqnarray*}
  where $a, b \in V$ and $H \cdot a$ is the orbit of $a$ and $H \cdot b$ is
  the orbit of $b$. What remains to show is that $w (a, b)$ is a feasible flow
  of $G$ and that the value of $w$ is the same as the value of $w'$. Firstly,
  note that $w / H$ fulfill constraints \eqref{eqA} and \eqref{eqB} from
  Definition \ref{def_maxflow} for the graph $G$ since the constraints are
  linear. For example take constraint \eqref{eqA},
  \begin{eqnarray*}
    w (a, b) + w (b, a) & = & (w' (H \cdot a, H \cdot b) + w' (H \cdot b, H
    \cdot a)) \frac{C (a, b)}{(C / H) (H \cdot a, H \cdot b)}\\
    & \leqslant & (C / H) (H \cdot a, H \cdot b)  \frac{C (a, b)}{(C / H) (H
    \cdot a, H \cdot b)}\\
    & = & C (a, b) .
  \end{eqnarray*}
  Secondly note that the value of $w$ is the same as $w'$ since
  \begin{eqnarray*}
    \val (w') & = & \sum_{A \in S / H} \fout_{w'} (A) -
    \fin_{w'} (A)\\
    & = & \sum_{A \in S / H} \sum_{B \in V / H} w' (A, B) - w' (B, A)\\
    & = & \sum_{A \in S / H} \sum_{B \in V / H} (w' (A, B) - w' (B, A))
    \left( \sum_{a \in A} \sum_{b \in B} \frac{C (a, b)}{(C / H) (A, B)}
    \right)\\
    & = & \sum_{A \in S / H} \sum_{B \in V / H} \sum_{a \in A} \sum_{b \in B}
    (w' (A, B) - w' (B, A)) \frac{C (a, b)}{(C / H) (A, B)}\\
    & = & \sum_{A \in S / H} \sum_{B \in V / H} \sum_{a \in A} \sum_{b \in B}
    w (a, b) - w (b, a)\\
    & = & \sum_{a \in S} \sum_{b \in V} w (a, b) - w (b, a)\\
    & = & \sum_{a \in S} \fout_w (a) - \fin_w (a)\\
    & = & \val (w) .
  \end{eqnarray*}
  So $w$ is a feasible flow of $G$ and $\val (w) = \val (w')$, so
  $\maxflowval (G) \geqslant \maxflowval (G / H)$.
\end{proof}

\section{Properties of relaxed soundness }\label{appendix:B}

The relaxed soundness share many similarities with the (true) soundness. One
example is the following Proposition, which is an analogue to Proposition
\ref{prop:3.7} but for relaxed soundness.

\begin{proposition}
  \label{prop:pi} For any \hadkttlint{} gadget $(G,
  \mathbb{X}_k, \mathbb{Y}_k)$

\begin{enumerate}[(a)]
\item
\begin{equation*}
s (G) \leqslant \rs (G) .
\end{equation*}
\item There exists a \hadkttlint{} gadget $(\tilde{G},
\mathbb{X}_k, \mathbb{Y}_k)$ with completeness $1 - 2^{- k}$ such that
\begin{equation*}
\frac{1 - \rs (G)}{1 - c (G)} \leqslant \frac{1 - \rs
(\tilde{G})}{1 - c (\tilde{G})},
\end{equation*}
\item and for any \hadkttlint{} gadget $(\tilde{G},
\mathbb{X}_k, \mathbb{Y}_k)$ with completeness $1 - 2^{- k}$
\begin{equation*}
\frac{1 - \rs (\tilde{G})}{1 - c (\tilde{G})} \leqslant 2.
\end{equation*}
\end{enumerate}
\end{proposition}

\begin{proof}
  \
  \begin{description}
  \item[(a)] Note that interpreting $x_1$ and $x_{- 1}$ as being primary
  variables do not affect soundness, i.e.
  \begin{eqnarray}
    s (G) & = & \underset{P \in \Ffold{} (\mathbb{X}_k \cup
    \{ x_1, x_{- 1} \})}{\expect} \max_{\begin{array}{l}
      A \in \Ffold{} (\mathbb{X}_k \cup \mathbb{Y}_k),\\
      A (x) = P (x), x \in \mathbb{X}_k \cup \{ x_1, x_{- 1} \}
    \end{array}} \val (A, G) .  \label{eq:sound4}
  \end{eqnarray}
  The reason for this is that there exists a degree of freedom in the choice
  of $A$ since for any $A$, $\val (A, G) = \val (1 + A, G)$. This
  means for example that we can add one extra constraint like $A (x_1) = 1 + A
  (x_{- 1}) = 1$ to the definition of $s (G)$ without affecting its value.
  
  Comparing \eqref{eq:sound4} and the definition of relaxed soundness, we can
  clearly see that $s (G) \leqslant \rs (G)$ since the relaxed soundness
  is a less constrained maximisation problem compared to the right hand side
  of \eqref{eq:sound4}.
  
  \item[(b)] This proof is analogous to the proof of {\cite[Proposition
  2.29]{Has2}}. Note that by definition $1 - c (\tilde{G})$ is the average
  length of edges $(f_1, f_2)$ of the gadget $\tilde{G}$, weighted by
  $\tilde{G} (f_1, f_2)$. For $\tilde{G}$ to have completeness $1 - 2^{- k}$,
  the edges in $\tilde{G}$ need to have an average length of $2^{- k} .$ Since
  there are no edges shorter than $2^{- k}$, $\tilde{G}$ can only put non-zero
  capacity on edges of length exactly $2^{- k}$.
  
  Construct $\tilde{G}$ using the following procedure. Start with $G$. Split
  up each edge $(f_1, f_2)$ in $G$ into an arbitrary path starting at $f_1$,
  ending at $f_2$, with edges of length $2^{- k}$, where the sum of lengths of
  edges in the path should be equal to the length of the original edge $(f_1,
  f_2)$. Remove the capacity of edge $(f_1, f_2)$ and give each edge in the
  path the same capacity as the capacity of the original edge $(f_1, f_2)$.
  This will increase the total capacity of the graph by a factor of $(1 - c
  (G)) / 2^k$. As a final step, normalize the capacity by dividing the
  capacity of all edges by $(1 - c (G)) / 2^k$. Let the resulting graph be
  $\tilde{G}$. Note that $\tilde{G}$ is a \hadkttlint{}
  consisting only of edges of length $2^{- k}$, so its completeness is $1 -
  2^{- k}$.
  
  Recall that $1 - \rs (G)$ can be interpreted as the expected value of
  a \maxflow{} problem on a fully connected $2^k$-dimensional hypercube, where
  the placements of sources and sinks have been randomised. Note that any
  feasible flow $\omega$ of $G$, when scaled down by a factor of $(1 - c (G))
  / 2^{- k}$, corresponds to a feasible flow of $\tilde{G}$. This implies that
  $(1 - \rs (G)) \leqslant (1 - \rs (\tilde{G})) (1 - c (G)) /
  2^{- k}$.
  
  The conclusion from this is that
  \begin{eqnarray*}
    \frac{1 - \rs (\tilde{G})}{1 - c (\tilde{G})} = \frac{1 - \rs
    (\tilde{G})}{2^{- k}} & \geqslant & \frac{1 - \rs (G)}{1 - c (G)} .
  \end{eqnarray*}
  \item[(c)] Let $\tilde{G}$ be the gadget from b). Recall that $1 -
  \rs (\tilde{G})$ can be interpreted as the expected value of a
  \maxflow{} problem on a fully connected $2^k$-dimensional hypercube, where the
  placements of sources and sinks have been randomised. The capacities of this
  flow graph sum to $1$.
  
  Note that the sources and sinks correspond to affine functions, which have a
  normalised Hamming distance of at least $1 / 2$. Furthermore, since all
  edges in $\tilde{G}$ has length $2^{- k}$, any path in $\tilde{G}$ between a
  source and a sink must contain at least $2^{k - 1}$ edges.
  
  For any flow graph, if all paths between sources and sinks contain at least
  $2^{k - 1}$ edges, and the sum of capacity over all edges in the graph is
  $1$, then the maximum flow is at most $2^{1 - k}$. So $1 - \rs
  (\tilde{G}) \leqslant 2^{1 - k}$, which implies that
  \begin{eqnarray*}
      \frac{1 - \rs (\tilde{G})}{1 - c (\tilde{G})} \leqslant \frac{2^{1
      - k}}{2^{- k}} = 2.
  \end{eqnarray*}
  \end{description}
\end{proof}

\begin{remark}
  Since the relaxed soundness is an upper bound of the true soundness, it
  follows that the \NPhardness{} result of \maxtlint{} as stated in
  Proposition \ref{prop:gap} also holds for $s = \rs (G) .$
\end{remark}

\section{Affine maps and lifts \label{sec:affine}\label{sec4}}

Recall that the $\rsLP(G)$ can be interpreted as the expected value of a
\maxflow{} problem with a randomised source/sink placement over a fully
connected $2^k$-dimensional hypercube, where the nodes are indexed by Boolean
functions $f \in \Fk_k$. The source/sink nodes are indexed by affine
Boolean functions. In order to be able to describe the symmetries of these
graphs, we want to study mappings $M : \Fk_k \rightarrow
\Fk_k$ with the following properties:
\begin{enumerate}
  \item Source and sink nodes map to source and sink nodes, i.e. if $f$ is an
  affine function then $M (f)$ is also an affine function.
  
  \item The length of all edges $\{ v_{f_1}, v_{f_2} \}$ are preserved by the
  mapping, i.e. $\dist (M (f_1), M (f_2)) = \dist (f_1, f_2)$.
\end{enumerate}
There is a natural choice of mappings from $\Fk_k \rightarrow
\Fk_k$ for which Property 1 and 2 hold. Additionally as a bonus, the
same natural choice of mappings can also be extended to construct mappings
from $\Fk_k \rightarrow \Fk_{k'}, k \leqslant k'$, and still
have that both Property 1 and 2 hold. This can then be used to embed the
$2^k$-dimensional hypercube in the $2^{k'}$-dimensional hypercube.

\begin{definition}
  \label{def:affinemap}Let $M_{A, b, \beta, c} : \Fk_k \rightarrow
  \Fk_{k'}$ be defined as
  \begin{eqnarray*}
    M_{A, b, \beta, c} (f) (y) & = & f (Ay + b)  (- 1)^c \chi_{\beta} (y),
  \end{eqnarray*}
  where $k, k' \in \integers_{> 0}$, $k \leqslant k'$, $y \in
  \bool{}^{k'}$, $A \in \bool{}^{k \times k'}$ is a full rank
  matrix, $b \in \boolk{}$, $c \in \bool{}$ and $\beta \in \boolkp$. Let
  $\Mk_{k \rightarrow k'}$ denote the set of all maps $M_{A, b, \beta,
  c}$ from $\Fk_k \rightarrow \Fk_{k'}$. For convenience, we
  often denote $M_{A, b, \beta, c}$ by $M$, where the $A, b, \beta, c$ are all
  implicit.
\end{definition}

Since these mappings are reminiscent of affine maps from linear algebra, we
call them \emph{affine maps}. However, they are not affine maps in the
classical sense.

The function $M (f) \in \Fk_{k'}$ is called the \emph{$M$-lift of
$f$}. It is not hard to see that the $M$-lift of an affine function is an
affine function. More generally, $M$-lifts always preserve the dimension of
Boolean functions.

\begin{proposition}
  \label{prop:dim}Given $f \in \Fk_k$ and $M \in \Mk_{k
  \rightarrow k'}$, $k \leqslant k'$, then $\dimension (M (f)) = \dimension (f)$.
\end{proposition}

\begin{proof}
  It follows from a direct calculation that
  \begin{eqnarray*}
    M_{A, b, \beta, c} (f) (y) & = & (- 1)^c \sum_{\alpha \in \{ 0, 1 \}^k}
    \chi_{A^T \alpha + \beta} (y)  \hat{f}_{\alpha} \chi_b (\alpha) .
  \end{eqnarray*}
  This shows that the affine mapping $M$ moves $\affine (f)$ to
  $\affine (M (f)) = \{ A^T \alpha + \beta : \alpha \in \affine
  (f) \}$. Furthermore, since $A$ is a full rank matrix, $\dimension (f) = \dimension (M
  (f))$. 
\end{proof}

Affine maps also preserve the (normalised Hamming) distance of affine
functions.

\begin{proposition}
  \label{prop:dist}Given $f_1, f_2 \in \Fk_k$ and $M \in
  \Mk_{k \rightarrow k'}$, $k \leqslant k'$, then $\dist (M
  (f_1), M (f_2)) = \dist (f_1, f_2)$.
\end{proposition}

\begin{proof}
  Let $M = M_{A, b, \beta, c}$. Note that $\dist (M (f_1), M (f_2))$ only depends on $A$ and $b$ since
  \begin{eqnarray*}
    \dist (M (f_1), M (f_2)) & = & \frac{1}{2^{k'}} \sum_{y \in \boolkp}
    \frac{1 - M (f_1) (y) M (f_2) (y)}{2}\\
    & = & \frac{1}{2^{k'}} \sum_{y \in \boolkp} \frac{1 - f_1 (Ay + b) f_2
    (Ay + b)}{2} .
  \end{eqnarray*}
  Furthermore, since $A$ is a full rank $k \times k'$ Boolean matrix, the
  kernel of $A$ has dimension $k' - k$ and size $2^{k' - k}$, so
  \begin{eqnarray*}
    \sum_{y \in \{ 0, 1 \}^{k'}} f_1 (Ay + b) f_2 (Ay + b) & = & 2^{k' - k}
    \sum_{x \in \{ 0, 1 \}^k} f_1 (x) f_2 (x) .
  \end{eqnarray*}
  This shows that $\dist (M (f_1), M (f_2)) = \dist (f_1, f_2)$.
\end{proof}

The last notable property of the affine maps is that they form a group under
composition. This property is needed to be able to apply the techniques from
Appendix \ref{sec:flowsym} to the $\rsLP(G)$ and to the $\rsinfLP(G)$ in
order to ``compress'' them.

\begin{proposition}
  \label{prop:group}$\Mk_{k \rightarrow k}$ under composition forms a group.
\end{proposition}

\begin{proof}
  The composition of two affine maps $M_{A', b', \beta', c'} \circ M_{A, b,
  \beta, c}$, is an affine map $M_{A'', b'', \beta'', c''}$, where
  \begin{eqnarray*}
    A'' & = & AA',\\
    b'' & = & Ab' + b,\\
    \beta'' & = & (A')^T \beta + \beta',\\
    c'' & = & (b', \beta) + c' + c.
  \end{eqnarray*}
  Furthermore, the left and right inverse of an affine map $M_{A, b, \beta, c}$
  is given by $M_{A', b', \beta', c'}$ where
  \begin{eqnarray*}
    A' & = & A^{- 1},\\
    b' & = & A^{- 1} b,\\
    \beta' & = & (A^{- 1})^T \beta,\\
    c' & = & c + (A^{- 1} b, \beta) .
  \end{eqnarray*}
  This shows that $\Mk_{k \rightarrow k}$ forms a group under
  composition.  
\end{proof}

\subsection{\texorpdfstring{$M$}{M}-lifts of sink and sources}

Recall that the source/sink placements of the $\rsLP(G)$ and the
$\rsinfLP(G)$ are described using a Boolean function $g \in \Fk_k$,
\begin{eqnarray*}
  g (\alpha) & = & \left\{\begin{array}{ll}
    1 & \text{iff } v_{\chi_{\alpha}}  \text{ is a sink},\\
    - 1 & \text{iff } v_{\chi_{\alpha}}  \text{ is a source.}
  \end{array}\right.
\end{eqnarray*}
Note that $M$-lifts move the sink and source nodes. If $k = k'$, then the
$M$-lift permutes the sink and source nodes. If $k < k'$, then the $M$-lift
``lifts'' the sink and source nodes onto a higher dimensional hypercube. This
means that there exists multiple different source/sink placements $g' \in
\Fk_{k'}$ that all match the lifted positions of the sinks and
sources. The condition for when an $M$-lift of a source/sink placement $g \in
\Fk_k$ is described by a source/sink placement $g' \in
\Fk_{k'}$ is given by the following proposition.

\begin{proposition}
  \label{prop_sinksourcelift}An $M$-lift will map sink nodes in
  $\Fk_{k}$ onto sink nodes of $\Fk_{k'}$ and source nodes in
  $\Fk_k$ onto source nodes in $\Fk_{k'}$ if and only if
  \begin{eqnarray*}
    M_{A^T, \beta, b, c} (g') & = & g.
  \end{eqnarray*}
\end{proposition}

\begin{proof}
  \
  
  Note that the $M_{A, b, \beta, c}$-lift of $\chi_{\alpha}$ is $M_{A, b,
  \beta, c} (\chi_{\alpha}) = \chi_{A^T \alpha + \beta} (x)  (- 1)^c \chi_b
  (\alpha)$. Using the source/sink placement $g'$ we can tell whether a node
  $v_{\chi_{\alpha}}$ is lifted onto a sink node or a source node,
  \begin{eqnarray*}
    g' (A^T \alpha + \beta)  (- 1)^c \chi_b (\alpha) & = &
    \left\{\begin{array}{ll}
      1 & \text{iff } v_{M_{A, b, \beta, c} (\chi_{\alpha})}  \text{ is a sink according to } g',\\
      - 1 & \text{iff } v_{M_{A, b, \beta, c} (\chi_{\alpha})}  \text{ is a
      source according to } g' .
    \end{array}\right.
  \end{eqnarray*}
  This implies that the sufficient and necessary condition to make all sinks
  in $\Fk_k$ to be $M_{A, b, \beta, c}$-lifted to sinks in
  $\Fk_{k'}$ and all sources in $\Fk_k$ to be $M_{A, b, \beta,
  c}$-lifted to sources in $\Fk_{k'}$, is that
  \begin{eqnarray*}
    g (\alpha) & = & g' (A^T \alpha + \beta) (- 1)^c \chi_b (\alpha)
  \end{eqnarray*}
  for all $\alpha \in \boolk$. This is identical to requiring that $M_{A^T,
  \beta, b, c} (g') = g$. 
\end{proof}

\begin{definition}
  The operator $M_{A^T, \beta, b, c}$ is denoted by $M^{\#}_{A, b, \beta, c}$.
\end{definition}

\subsection{Lifting gadgets and flows}

It is possible to extend the definition of $M$-lifting to \hadkttlint{} gadgets $G$ by defining $M \cdot G$ as
\begin{eqnarray*}
  (M \cdot G) (f_1', f_2') & = & \sum_{\begin{array}{l}
    f_1 \in M^{- 1} (f'_1),\\
    f_2 \in M^{- 1} (f'_2)
  \end{array}} G (f_1, f_2) .
\end{eqnarray*}
This moves the capacity $G (f_1, f_2)$ of edge $\{ v_{f_1}, v_{f_2} \}$ onto
edge $\{ v_{M (f_1)}, v_{M (f_2)} \}$. Furthermore, let the \emph{full $k \rightarrow k'$ lift} of $G$ be defined as the average of all
possible $M$-lifts, i.e.
\begin{eqnarray*}
  \lift_{k \rightarrow k'} (G) & = & \frac{1}{| \Mk_{k
  \rightarrow k'} |} \sum_{M \in \Mk_{k \rightarrow k'}} (M \cdot G) .
\end{eqnarray*}
Completely analogue to the definition of $M$-lifts of gadgets, let the
$M$-lift of a flow $w$ of the $\rs (G)$ LP be defined as
\begin{eqnarray*}
  (M \cdot w) (f'_1, f_2', g') & = & \sum_{\begin{array}{l}
    f_1 \in M^{- 1} (f'_1),\\
    f_2 \in M^{- 1} (f'_2)
  \end{array}} w (f_1, f_2, M^{\#} (g')),
\end{eqnarray*}
and let the full $k \rightarrow k'$ lift of $w$ be defined as
\begin{eqnarray*}
  \lift_{k \rightarrow k'} (w) & = & \frac{1}{| \Mk_{k
  \rightarrow k'} |} \sum_{M \in \Mk_{k \rightarrow k'}} (M \cdot w) .
\end{eqnarray*}
By connecting these two concepts of lifting gadgets and flows, we can show
the following proposition.

\begin{proposition}
  \label{prop:Gprim}The full lift of $G$ is a \hadkttlint{} gadget $G'$ where $c (G') = c (G)$ and $\rs (G') \leqslant
  \rs (G)$.
\end{proposition}

\begin{proof}
  Let $w$ be a feasible flow of $G$ and let $w'=\lift_{k \rightarrow
  k'} (w)$. Note that $w'$ is a feasible flow of $G'$ since the capacity of
  $G$ is lifted together with the flow $w$. So constraints $\eqref{eq1}$ and
  $\eqref{eq2}$ are satisfied by $w'$. Additionally,
  \begin{eqnarray*}
    \expect_{g \in \Fk_k} \val_g (w) & = & \expect_{g' \in
    \Fk_{k'}} \val_{g'} (w') .
  \end{eqnarray*}
  since any lift preserves the amount of flow going in and out of sink nodes and source nodes.
\end{proof}

The final Proposition that we need for Appendix \ref{sec_proof} is that the
full lift of a leaky flow $w$ of the $\rsLP(G)$ is a leaky flow of the
$\rsLP(G')$, and that the full lift does not affect the value of the flow.
This is a fundamental property of lifts that is used in Appendix
\ref{sec_proof} to upper bound $\rs (G')$ when $k' \rightarrow \infty$.

\begin{proposition}
  \label{prop:leaky}Let $G'$ be the full lift of $G$, and let $w'$ be the full
  lift of a leaky flow $w$ of the $\rsLP(G)$. Then $w'$ is a leaky
  flow of the $\rsLP (G')$, and $\expect_{g \in \Fk_k}
  \val_g (w) = \expect_{g' \in \Fk_{k'}} \val_{g'}
  (w')$.
\end{proposition}

\begin{proof}
  Let $w$ be a leaky flow of $G$ and let $w'=\lift_{k \rightarrow k'}
  (w)$. Note that constraint $\eqref{eq1}$ is satisfied by $w'$ since the
  capacity of $G$ is lifted together with the flow $w$. So $w'$ is a leaky
  flow. Additionally,
  \begin{eqnarray*}
    \expect_{g \in \Fk_k} \val_g (w) & = & \expect_{g' \in
    \Fk_{k'}} \val_{g'} (w') .
  \end{eqnarray*}
  since any lift preserves the amount of flow going in and out of sink nodes
  and source nodes.
\end{proof}

\section{Proving that \texorpdfstring{$\rsinf(G)$}{rsinf(G)} can be attained in the limit
\label{sec_proof}}

The goal of this section is to prove Lemma \ref{thm_important}, which relates
the infinity relaxed soundness to the relaxed soundness. Let $G$ be the
\hadkttlint{} gadget in Lemma \ref{thm_important} and
let $w$ be the optimal flow of the $\rsinfLP(G)$, which implies that
$\rsinf (G) = 1 - \expect_{g \in \Fk_k} \val_g
(w)$. Let $k'$ be some integer greater than $k$ and define $G' =
\lift_{k \rightarrow k'} (G)$ and $w' = \lift_{k \rightarrow k'}
(w)$. According to Proposition \ref{prop:Gprim} $G'$ is a
\hadkttlint{} with $c (G') = c (G)$ and according to
Proposition \ref{prop:leaky} $w'$ is a leaky flow of the $\rsLP (G')$ and
$\expect_{g \in \Fk_k} \val_g (w) = \expect_{g' \in
\Fk_{k'}} \val_{g'} (w')$. We prove that as $k'$ tends to
infinity the total leakage of $G'$ converges to $0$. After we have established
this, Lemma \ref{thm_important} follows from Theorem \ref{leaky}.

\subsection{Total leakage approaches \texorpdfstring{$0$ as $k' \rightarrow \infty$}{0 as k' -> inf}}

Let us start by formally defining the leaks of $w$ and $w'$, where $w$ is a
leaky flow of the $\rsLP (G)$ and $w'$ is a leaky flow of the $\rsLP (G')$.
Recall that the $\rsLP(G')$  describe the expectation of the maximum flow of a
graph with a random source/sink placement $g' \in \Fk_{k'}$. It is for
this reason that the total leakage of $w'$ is defined as an expectation over
$g' \in \Fk_{k'}$ of the total leakage of the graph with source/sink
placement given by $g'$.

\begin{definition}
  Let $L_{k'}$ denote the total leakage of $w'$, \
  \begin{eqnarray*}
    L_{k'} & = & \expect_{g' \in \Fk_{k'}} \left(
    \sum_{\begin{array}{l}
      f' \in \Fk_{k'}\\
      s.t. \dimension (f') > 0
    \end{array}} | \leak_{w'} (f', g')  | \right),
  \end{eqnarray*}
  where
  \begin{eqnarray*}
    \leak_{w'} (f', g') & = & \fout_{w'} (f', g') - \fin_{w'}
    (f', g')\\
    & = & \frac{1}{| \Mk_{k \rightarrow k'} |} \sum_{M \in
    \Mk_{k \rightarrow k'}} \left( \sum_{\begin{array}{l}
      f \in \Fk_k\\
      s.t. M (f) = f'
    \end{array}} \leak_w (f, M^{\#} (g')) \right) .
  \end{eqnarray*}
\end{definition}

The aim of this subsection is to prove that $L_{k'} \rightarrow 0$ as $k'
\rightarrow \infty$. We do this by proving the following upper bound on
$L_{k'}$ through a second order moment analysis.

\begin{proposition}
  \label{prop_goal}
  \begin{eqnarray*}
    L_{k'} & \leqslant & \frac{{2^{2^k + k}} }{\sqrt{2^{k'} - 2^k}} .
  \end{eqnarray*}
\end{proposition}

The proof of Proposition \ref{prop_goal} relies on the following Proposition
describing the relationship between random pairs of affine maps $M_1, M_2 \in
\Mk_{k \rightarrow k'}$ such that $M_1 (f) = M_2 (f)$ for some fixed
$f \in \Fk_k$.

\begin{definition}
  Given $M_{A, b, \beta, c} \in \Mk_{k \rightarrow k'}$, let $T_M :
  \boolk \rightarrow \boolkp$ denote the affine map $T_M (x) = A^T x + \beta$.
  Furthermore, let $\affine (M_{A, b, \beta, c})$ denote the affine
  subspace $\left\{ T_M (x)  : x \in \boolk \right\} {\subseteq
  \real^{k'}} $.
\end{definition}

\begin{proposition}
  \label{prop_pair}Given $f \in \Fk_k$ and $f' \in \Fk_k$ with
  $\dimension (f) = \dimension (f') = d$. Then
   \[
    | \{ (M_1, M_2) \in \Nk_{f \rightarrow f'} \times \Nk_{f
    \rightarrow f'} : \dimension (\affine (M_1) \cap \affine (M_2)) > d
    \} | \leqslant | \Nk_{f \rightarrow f'} |^2 \frac{\left( {2^k}
    - 2^d \right)^2}{2^{k'} - 2^d},
  \]
  where $\Nk_{f \rightarrow f'} = \{ M \in \Mk_{k \rightarrow
  k'} : M (f) = f' \}$ denotes the set of affine maps in $\Mk_{k
  \rightarrow k'}$ that lifts $f$ to $f'$.
\end{proposition}

\begin{proof}
  Note that for any $M_1, M_2 \in \Nk_{f \rightarrow f'}$, the
  dimension of $\affine (M_1) \cap \affine (M_2)$ is at least $d$,
  since according to the proof of Proposition \ref{prop:dim} both $T_{M_1}$
  and $T_{M_2}$ must map $\affine (f)$ onto $\affine (f')$, so
  $\dimension (\affine (M_1) \cap \affine (M_2) \cap \affine (f'))
  = d$. However, the two maps $T_{M_1}$ and $T_{M_2}$ can map the complement
  of $\affine (f)$ in different ways since there is no restriction to
  how they map the complement of $\affine (f)$.
  
  Fix $M_1$ and uniformly at random pick $M_2$ from $\Nk_{f
  \rightarrow f'} .$ Given any fix $x \nin \affine (f)$, the probability
  that $T_{M_2} (x) \in \affine (M_1)$ is $(2^k - 2^d) / (2^{k'} - 2^d)$
  since $| \affine (M_1) \setminus \affine (f') | = 2^k - 2^d$ and
  $T_{M_2} (x)$ is \ uniformly distributed over the complement of
  $\affine (f')$. Taking a union bound over all $x \nin \affine (f)$
  shows that
  \begin{eqnarray*}
    P_{M_2 \in \Nk_{f \rightarrow f'}} 
    [\dimension (\affine (M_1) \cap
    \affine (M_2)) > d] & \leqslant & \frac{\left( {2^k}  - 2^d
    \right)^2}{2^{k'} - 2^d} .
  \end{eqnarray*}
  Proposition \ref{prop_pair} follows directly from this inequality.
\end{proof}

The takeaway from Proposition \ref{prop_pair} is that if $M_1$ and $M_2$ are
two random affine maps such that $M_1 (f) = M_2 (f)$ for some fixed $f \in
\Fk_k$, then with high probability $\affine (M_1) \cap
\affine (M_2) = \affine (f)$. This allows us to create a bound on
the second order moment of the terms that define $L_{k'}$.

\begin{lemma}
  \label{lemma_key}Given $f \in \Fk_k$, $f' \in \Fk_{k'}$ and
  $g' \in \Fk_{k'}$, where $\dimension (f) = \dimension (f') = d > 0$, then
  \begin{eqnarray*}
    \expect_{g' \in \Fk_{k'}} \left( \left|
    \sum_{\begin{array}{l}
      M \in \Nk_{f \rightarrow f'}
    \end{array}} \leak_w (f, M^{\#} (g')) \right|^2 \right) & \leqslant
    & | \Nk_{f \rightarrow f'} |^2 \frac{\left( {2^k}  - 2^d
    \right)^2}{2^{k'} - 2^d} .
  \end{eqnarray*}
\end{lemma}

\begin{proof}
  Expanding the square we need to prove that,
  \[
    \sum_{\begin{array}{l}
      M_1, M_2 \in \Nk_{f \rightarrow f'}
    \end{array}} \expect_{g' \in \Fk_{k'}}  (\leak_w (f,
    M_1^{\#} (g')) \leak_w (f, M_2^{\#} (g'))) \leqslant  |
    \Nk_{f \rightarrow f'} |^2 \frac{\left( {2^k}  - 2^d
    \right)^2}{2^{k'} - 2^d} .
  \]
  Split the terms up into two cases, either $\dimension (\affine (M_1) \cap
  \affine (M_2)) > d$ or $\dimension (\affine (M_1) \cap \affine
  (M_2)) = d$. By Proposition \ref{prop_pair} the number of terms of the first
  type is at most $| \Nk_{f \rightarrow f'} |^2  \left( {2^k}  - 2^d
  \right)^2 / (2^{k'} - 2^d)$. Each term is bounded by one since the sum of
  capacities in the $\rs (G)$ LP is equal to $1$, so the absolute value
  of a leak is always smaller than or equal to $1$ at any node and for any
  source/sink placement.
  
  In the other case, when $\dimension (\affine (M_1) \cap \affine (M_2))
  = d$, then the two random functions $M_1^{\#} (g')$ and $M_2^{\#} (g')$ are
  equal on $\affine (f)$, and independently uniformly random $\{ 1, - 1
  \}$ on the complement of $\affine (f)$. This allows us to rewrite the
  expectation over $g'$ as
  \begin{equation*}
  \begin{split}
    & \expect_{g' \in \Fk_{k'}} (\leak_w (f, M^{\#}_1
    (g')) \leak_w (f, M^{\#}_2 (g')))\\
    = & \expect_{g' \in \Fk_{k'}} \left( \leak_w (f,
    M^{\#}_1 (g')) \expect_{\begin{array}{l}
      g_2' \in \Fk_{k'}\\
      s.t. M^{\#}_2 (g_2') |_{\affine (f)} = M^{\#}_1 (g') 
      |_{\affine (f)}
    \end{array}} \leak_w (f, M^{\#}_2 (g'_2)) \right)\\
    = & \expect_{g \in \Fk_k} \left( \leak_w (f, g)
    \expect_{\begin{array}{l}
      g_2 \in \Fk_k\\
      s.t. g_2 |_{\affine (f)} = g  |_{\affine (f)}
    \end{array}} \leak_w (f, g_2) \right) .
  \end{split}
  \end{equation*}

    This is equal to 0, since for any
    infinity relaxed flow $w$ (see
    Definition
    \ref{def:relaxed_constraint}) the
    expectation of $\leak_w (f,
    g_2)$ over $g_2$ given $g$ is 0.
\end{proof}

We are now at the point where we can prove Proposition \ref{prop_goal} using
Lemma \ref{lemma_key}.

\begin{proof}[Proof of Proposition \ref{prop_goal}]
  A trivial upper bound of $L_{k'}$ using the triangle inequality is
  \begin{eqnarray*}
    L_{k'} & \leqslant & \frac{1}{| \Mk_{k \rightarrow k'} |} \sum_{f
    \in \Fk_k} \sum_{\begin{array}{l}
      f' \in \Fk_{k'}\\
      s.t. \dimension (f') > 0
    \end{array}} \expect_{g' \in \Fk_{k'}} \left( \left|
    \sum_{
      M \in \Nk_{f \rightarrow f'}
    } \leak_w (f, M^{\#} (g')) \right| \right) .
  \end{eqnarray*}
  Applying Jensen's inequality to the expectation over $g' \in
  \Fk_{k'}$ gives
  \[
    \expect_{g' \in \Fk_{k'}} \left( \left|
    \sum_{
      M \in \Nk_{f \rightarrow f'}
    } {\leak_w}  (f, M^{\#} (g')) \right| \right) 
    \leqslant \sqrt{\expect_{g' \in \Fk_{k'}} \left( \left|
    \sum_{
      M \in \Nk_{f \rightarrow f'}
    } \leak_w (f, M^{\#} (g')) \right|^2 \right)},
  \]
  which according to to Lemma \ref{lemma_key} can be further upper bounded by
  \begin{eqnarray*}
    \sqrt{\expect_{g' \in \Fk_{k'}} \left( \left|
    \sum_{
      M \in \Nk_{f \rightarrow f'}
    } \leak_w (f, M^{\#} (g')) \right|^2 \right)} &
    \leqslant & \frac{{2^k}  - 2^{\dimension (f)}}{\sqrt{2^{k'} - 2^{\dimension (f)}}} |
    \Nk_{f \rightarrow f'} |\\
    & \leqslant & \frac{{2^k} }{\sqrt{2^{k'} - 2^k}} | \Nk_{f
    \rightarrow f'} | .
  \end{eqnarray*}
  We have so far shown that
  \begin{eqnarray*}
    L_{k'} & \leqslant & \frac{{2^k} }{\sqrt{2^{k'} - 2^k}} \sum_{f \in
    \Fk_k} \sum_{\begin{array}{l}
      f' \in \Fk_{k'}\\
      s.t. \dimension (f') > 0
    \end{array}} \frac{| \Nk_{f \rightarrow f'} |}{| \Mk_{k
    \rightarrow k'} |} .
  \end{eqnarray*}
  Finally, note that $\sum_{f' \in \Fk_{k'}} | \Nk_{f
  \rightarrow f'} | = | \Mk_{k \rightarrow k'} |$ since
  $\Nk_{f \rightarrow f'}$ are disjoint subsets of $\Mk_{k
  \rightarrow k'}$ for different $f' \in \Fk_{k'}$ and their union
  over $f' \in \Fk_{k'}$ is equal to $\Mk_{k \rightarrow k'}$.
  So
  \begin{eqnarray*}
    L_{k'} & \leqslant & \frac{{2^k} }{\sqrt{2^{k'} - 2^k}} \sum_{f \in
    \Fk_k} \sum_{\begin{array}{l}
      f' \in \Fk_{k'}\\
      s.t. \dimension (f') > 0
    \end{array}} \frac{| \Nk_{f \rightarrow f'} |}{| \Mk_{k
    \rightarrow k'} |} \leqslant \frac{{2^k} }{\sqrt{2^{k'} - 2^k}} \sum_{f
    \in \Fk_k} 1 \leqslant \frac{{2^{2^k + k}} }{\sqrt{2^{k'} - 2^k}}.
  \end{eqnarray*}
  
\end{proof}

\subsection{The proof of Lemma \ref{thm_important}}

All that remains is to tie up the loose ends by proving Lemma
\ref{thm_important} using Proposition \ref{prop_goal} combined with Theorem
\ref{leaky}.

\begin{proof}[Proof of Lemma \ref{thm_important}]
  Since $w'$ is a leaky flow of the $\rsLP (G')$, it follows from Theorem
  \ref{leaky} that there exists a feasible flow $\tilde{w}'$ of the $\rsLP
  (G')$ such that
  \begin{eqnarray*}
    \expect_{g' \in \Fk_{k'}} \val_{g'} (\tilde{w}') + L_{k'}
    & \geqslant & \expect_{g' \in \Fk_{k'}} \val_{g'} (w') .
  \end{eqnarray*}
  Note that $\rs(G') \geqslant 1 - \expect_{g' \in \Fk_{k'}}
  \val_{g'} (\tilde{w}')$ since $\tilde{w}'$ is a feasible flow of the
  $\rsLP (G')$. Furthermore, recall that $\rsinf (G) = 1
  - \expect_{g' \in \Fk_{k'}} \val_{g'} (w')$. So
  \begin{eqnarray*}
    \rs (G') - L_{k'} & \leqslant & \rsinf (G) .
  \end{eqnarray*}
  Proposition \ref{prop_goal} implies that $L_{k'} \rightarrow 0$ as $k'
  \rightarrow \infty$, which proves that $\forall \varepsilon > 0$ there
  exists a gadget $G'$ with $c (G') = c (G)$ such that $\rs (G') -
  \varepsilon \leqslant \rsinf (G)$.
\end{proof}

\section{Gadget construction and verification \label{sec_comp}}

This section contains the details for how to practically compute
\hadkttlint{} gadgets using the $\rsLP(G)$ and the
$\rsinfLP(G)$ . These LPs have far too many variables and constraints to
directly be solved by a computer when $k \geqslant 4$. The solution is to make
use of the symmetries of the LP:s to construct smaller LP:s with the same
optimum. This is done in two steps. Step 1 is to use Proposition
\ref{prop:Gprim} to argue that best gadgets are the symmetrical gadgets. This
means that we only need to take into account symmetrical gadgets when solving
the $\rsLP (G)$ and the $\rsinfLP (G)$. Step 2 is to use the fact that if $G$
is symmetrical, then Theorem \ref{thm_symmax} allows us to compress the LP,
merging a huge number of variables into a single variable.

\subsection{Symmetrical \texorpdfstring{\hadkttlint{}}{Hadk-to-2Lin(2)} gadgets are
optimal}

The meaning of a \hadkttlint{} gadget $(G, \mathbb{X},
\mathbb{Y})$ being \emph{optimal} is that there exists no
\hadkttlint{} gadget $(\tilde{G}, \mathbb{X},
\mathbb{Y})$ such that $c (G) = c (\tilde{G})$ and $\rs (G) > \rs
(\tilde{G})$. The following Proposition states that symmetric gadgets are
optimal. By symmetric, we refer to the property that the gadget $G$ is
invariant under $M$-lifts.

\begin{proposition}
  \label{prop_sym}Given any \hadkttlint{} gadget $(G,
  \mathbb{X}, \mathbb{Y})$, there exists a symmetric \hadkttlint{} gadget $(\tilde{G}, \mathbb{X}, \mathbb{Y})$ such that $c
  (G) = c (\tilde{G})$ and $\rs (G) \geqslant \rs (\tilde{G})$.
\end{proposition}

\begin{proof}
  Let $\tilde{G} = \text{lift}_{k \rightarrow k} (G)$. According to
  Proposition \ref{prop:Gprim}, $c (G) = c (\tilde{G})$ and $\rs (G)
  \geqslant \rs (\tilde{G})$. Furthermore, $\tilde{G}$ is a symmetric
  gadget since for any $f_1, f_2 \in \Fk_k$ and $M \in \Mk_{k
  \rightarrow k}$,
  \begin{eqnarray*}
    (M \cdot \tilde{G}) (f_1, f_2) & = & \frac{1}{| \Mk_{k
    \rightarrow k} |} \sum_{M_2 \in \Mk_{k \rightarrow k}} ((M \circ
    M_2) \cdot \tilde{G}) (f_1, f_2)\\
    & = & \frac{1}{| \Mk_{k \rightarrow k} |} \sum_{M_2
    \in M \circ \Mk_{k \rightarrow k}} (M_2 \cdot \tilde{G}) (f_1,
    f_2) .
  \end{eqnarray*}
  According to Proposition \ref{prop:group}, $\Mk_{k \rightarrow k}$
  forms a group, so $M \circ \Mk_{k \rightarrow k} = \Mk_{k
  \rightarrow k}$. We have shown that $M \cdot \tilde{G} = \tilde{G}$ and thus
  $\tilde{G}$ is a symmetric gadget.
\end{proof}

\subsection{Compressing the \texorpdfstring{$\rsLP (G)$ and $\rsinfLP (G)$}{rsLP(G) and rsinfLP(G)} }\label{sub_comp}

As discussed earlier, both the $\rsLP (G)$ and the $\rsinfLP (G)$ can be
interpreted as \maxflow{} problems. Furthermore, if $G$ is symmetric under
$M$-lifts, then $\Mk_{k \rightarrow k}$ is a symmetry group for both
of these \maxflow{} problems. This means that we can apply Theorem
\ref{thm_symmax} to compress the \maxflow{} problems, giving us the
\emph{compressed $\rsLP (G)$} and the \emph{compressed $\rsinfLP
(G)$}.

One of the symmetries that the compression is able to capture is that many
different source/sink placements are equivalent. In a sense, the source/sink
placements of the compressed LPs consist of one representative source/sink
placement from each set of equivalent source/sink placements. This symmetry
turns out to be the main contributor as to why the compressed LP is
significantly smaller than the original LP.

Without the compression, the LPs each have $2^{3 \cdot 2^k}$ variables, which
for $k \geqslant 4$ is computationally infeasible. However, even with the
compression, for $k = 4$ the LPs are still large enough that it is
computationally challenging to solve them.

\subsubsection{Further restricting the compressed LPs }\label{sec:rest}

To further restrict the size of the LPs in the case of $k=4$, we heuristically identify a list of
beneficial gadget variables by solving the compressed LPs with floating point
numbers using Gurobi. Any gadget variable that is given non-zero weight in at
least one floating point solution is added to the list. Using this list, we
define the \emph{restricted compressed LP} as the compressed LP but with
all other gadget variables that are not on the list, removed. The list we use
can be found in Table \ref{fig13} in Appendix \ref{sec:used}. Note that one possible drawback to
restricting the LPs like this is that the restriction could lead to
construction of sub-optimal gadgets.

\cref{tab1,tab2,tab3} show the sizes of the LPs depending on if
compression or restriction is being applied. Note that the restricted and
compressed LP:s have significantly fewer variables than the original LP:s.

\begin{table}
\centering
  \begin{tabular}{|l|lll|lll|}
    \hline
    \  & \multicolumn{3}{c|}{$\rsLP (G)$} & \multicolumn{3}{c|}{$\rsinfLP (G)$} \\
    \hline
    Original & $163$ & $343$ & $534$ & $163$ & $343$ & $534$\\
    Compressed & $23$ & $38$ & $106$ & $23$ & $38$ & $106$\\
    \hline
  \end{tabular}
  \caption{Sizes of the $\rsLP (G)$ and $\rsinfLP (G)$ for
  \hadkttlint[2]{} gadgets $G$. The three numbers are the
  number of linear constraints, number of variables and number of non-zero
  entries in the constraints. All variables have the implicit constraint of being non-negative. \label{tab1}}
\end{table}

\begin{table}
\centering
  \begin{tabular}{|l|lll|lll|}
    \hline
    \  & \multicolumn{3}{c|}{$\rsLP (G)$} & \multicolumn{3}{c|}{$\rsinfLP (G)$} \\
    \hline
    Original & $8 \cdot 10^6$ & $2 \cdot 10^7$ & $5 \cdot 10^7$ & $8 \cdot
    10^6$ & $2 \cdot 10^7$ & $5 \cdot 10^7$\\
    Compressed & $298$ & $546$ & $2330$ & 243  & 462 & 1987\\
    \hline
  \end{tabular}
  \caption{Sizes of the $\rsLP (G)$ and $\rsinfLP (G)$ for
  \hadkttlint[3]{} gadgets $G$. The three numbers are the
  number of linear constraints, number of variables and number of non-zero
  entries in the constraints. All variables have the implicit constraint of being non-negative. \label{tab2}}
\end{table}

\begin{table}
\centering
  \begin{tabular}{|l|lll|lll|}
    \hline
    \  & \multicolumn{3}{c|}{$\rsLP (G)$} & \multicolumn{3}{c|}{$\rsinfLP (G)$} \\
    \hline
    Original & $1 \cdot 10^{14}$ & $3 \cdot 10^{14}$ & $4 \cdot 10^{14}$ & $1
    \cdot 10^{14}$ & $3 \cdot 10^{14}$ & $4{\cdot}10^{14}$\\
    Restricted & $2 \cdot 10^{11}$ & $4 \cdot 10^{11}$ & $6 \cdot 10^{11}$ &
    $2 \cdot 10^{11}$ & $4 \cdot 10^{11}$ & $6 \cdot 10^{11}$\\
    Compressed & $4 \cdot 10^5$ & $7 \cdot 10^5$ & $1 \cdot 10^7$ & $3 \cdot
    10^5$ & $6 \cdot 10^5$ & $9 \cdot 10^6$\\
    Restricted \& compressed & $3 \cdot 10^4$ & $6 \cdot 10^4$ & $2 \cdot
    10^5$ & $3 \cdot 10^4$ & $5 \cdot 10^4$ & $2 \cdot 10^5$ \\
    \hline
  \end{tabular}
  \caption{Sizes of the $\rsLP (G)$ and $\rsinfLP (G)$ for
  \hadkttlint[4]{}. The three numbers are the
  number of linear constraints, number of variables and number of non-zero
  entries in the constraints. \label{tab3}}
\end{table}

There is a special case where we do not need the restrictions. If the
completeness of a gadget is $1 - 2^{- k}$, then the gadget only has non-zero
weight on edges of length $2^{- k}$. There are comparatively relatively few edges of length
$2^{- k}$. This allows us to directly construct the gadget by solving the
non-restricted LP. So in the case of completeness $1 - 2^{- k}$, the gadgets
we construct are guaranteed to be optimal since we do not make use of any
restrictions.

\subsection{Implementation details}

The compressed $\rsLP (G)$ and compressed $\rsinfLP (G)$ are constructed using
a Python script where all of the calculations are done using integer
arithmetic. The script makes use of affine maps to efficiently compute the
symmetries of the two LPs, in order to compress them. The time and memory
complexities of the script are roughly $O (2^{2 \cdot 2^k})$, so the script is
able to handle $k = 2$, $3$ and $4$. In theory it would be possible to also
make the script support $k = 5$, but that would require both more powerful
hardware, as well as improving the time complexity to roughly $O (2^{2^k})$
time.

After having computed the compressed $\rsLP (G)$ and compressed $\rsinfLP
(G)$, the list of beneficial gadget variables found in Section \ref{list} are
used to construct the restricted compressed LPs. In order to solve the
compressed LP we use the exact rational number LP solver
QSopt\_ex{\cite{LPSOLVER}}. This results in a gadget described only using
rational numbers, as well as an accompanying compressed flow, also described
only using rational numbers.

\subsection{Verification of \texorpdfstring{$\rs (G)$ and $\rsinf (G)$}{rs(G) and rsinf(G)}}

It is significantly simpler to verify the relaxed soundness and the infinity
relaxed soundness of a gadget than it is to construct the gadget. The
verification can be done almost directly on the original LPs, without needing
the restricted compressed LPs or the compressed LPs.

The input to the verification program is a gadget $G :
\binom{\Fk_k}{2} \rightarrow [0, 1]$ together with a flow $w_g :
\Fk_k \times \Fk_k \rightarrow \real$, for each
source/sink placement equivalence class representative $g$. The flow acts as a
witness for the relaxed soundness / infinity relaxed soundness of the gadget.
In order to avoid floating point errors, we require both $G$ and the $w_g$ to
be rational.

The verification process is done in five steps.
\begin{enumerate}
  \item For each source/sink placement representative $g$, verify that the
  flow $w_g$ satisfies the capacity constraints of the $\rs (G)$ LP /
  $\rsinf (G)$ LP, i.e. that $w_g (f_1, f_2) + w_g (f_2, f_1)
  \leqslant G (f_1, f_2)$ for all $f_1, f_2 \in \Fk_k$.
  
  \item Verify that the gadget $G$ is symmetric under action by $M \in
  \Mk_{k \rightarrow k} $, meaning that for all functions $f_1, f_2
  \in \Fk_k$ and affine maps $M \in \Mk_{k \rightarrow k}$, it
  holds that $G (f_1, f_2) = G (M (f_1), M (f_2))$.
  
  \item For each source/sink placement representative $g$ and each function $f
  \in \Fk_k$, compute $\fin (f, g)$ and $\fout (f, g)$. Now
  extend $\fin$ and $\fout$ to be defined for all $f$ and
  $g$ in $\Fk_k$. For any source/sink placements $\tilde{g} \in
  \Fk_k$ that is not a representative, pick a map $M \in
  \Mk_{k \rightarrow k}$ and representative $g$ such that $g = M^{\#}
  (\tilde{g})$, and define $\fin (f, \tilde{g})$ as $\fin (M^{- 1}
  (f), g)$ and $\fout (f, \tilde{g})$ as $\fout (M^{- 1} (f), g)$.
  
  \item Verify the conservation of flow constraint in the $\rsLP (G)$ /
  $\rsinfLP (G')$ by iterating over all $(f, g) \in \Fk_k \times
  \Fk_k$ that are not sinks or sources. For the $\rsLP(G)$  this just
  involves checking that $\fin (f, g) = \fout (f, g)$. For the
  $\rsinfLP (G)$ this involves checking that $\sum_{g'} \fin (f, g') =
  \sum_{g'} \fout (f, g')$, where the sum is over all $g'$ such that $g'
  \barsuchthat_{\affine (f)} = g \barsuchthat_{\affine (f)}$.
  
  \item Compute and output the completeness and $\rs$ /
  $\rsinf$ of the gadget using the extended inflow and outflow as
  a witness.
\end{enumerate}
Note that the first step verifies the capacity constraints only for
representatives of equivalent source/sink placements. The second step checks
that the gadget $G$ is symmetric, which combined with the first step implies
that any extension of the flow to an arbitrary source/sink placement will
fulfil the capacity constraints. The fourth step checks that the conservation
of flow constraint is fulfilled, which in the case of the $\rsinfLP (G)$
involves computing the affine support of all possible source/sink placements.

The LP's we use and the gadgets we present in this paper can be found at
\url{https://github.com/bjorn-martinsson/NP-hardness-of-Max-2Lin-2}, 
as
well as a stand alone implementation of a verification script written in
Python. As described in the verification process above, the verification
requires a flow $w_g$ as input. So on the Github, there is also a script used
to generate this witness flow. This is done by solving the restricted
compressed $\rsLP(G)$  / $\rsinfLP (G)$ using an integral \maxflow{} solver, and
then uncompressing the result.

\section{Edges used/unused in constructed gadgets}\label{sec:used}

During the numerical analysis, we solve LPs to construct the gadgets. A gadget can be interpreted as a probability distribution over (undirected) edges.
\cref{fig11,fig12,fig13} list all edges that have been given non-zero weight in at least one solution to an LP, for $k = 2, 3,
4$. Recall that every gadget that we construct is symmetrical under the mappings of
$\Mk_{k \rightarrow k}$, so edges from the same edge orbit share the
same capacity. More specifically, the tables contain a list of all edge
orbits that are used in at least one constructed gadget.

\begin{table}
\centering
  $\begin{array}{llll}
    f_1 & f_2 & \mathrm{Ham} . \mathrm{dist} . & \mathrm{size}\\
    \hline
    0000 & 1000 & 1 & 32\\
    0000 & 1100 & 2 & 24
  \end{array}$
  \caption{The relevant edge orbits for \hadkttlint[2]{}
  gadgets. The edges of a \hadkttlint[2]{} gadget has a
  total of $4$ edge orbits, but only two are ever used in our constructed
  gadgets. The rest of the edges were always given capacity $0$ by the
  (rational) LP-solver. \label{fig11}}

  $\begin{array}{llll}
    f_1 & f_2 & \mathrm{Ham} . \mathrm{dist} . & \mathrm{Size}\\
    \hline
    00000000 & 10000000 & 1 & 128\\
    10000000 & 11000000 & 1 & 896\\
    00000000 & 11000000 & 2 & 448\\
    00000000 & 11110000 & 4 & 112
  \end{array}$
  \caption{The relevant edge orbits for \hadkttlint[3]{}
  gadget. The edges of a \hadkttlint[3]{} gadget has a
  total of 26 edge orbits, but only four are ever used in our constructed
  gadgets. The rest of the edges were always given capacity $0$ by the
  (rational) LP-solver.\label{fig12}}

  $\begin{array}{llll}
    f_1 & f_2 & \mathrm{Ham} . \mathrm{dist} . & \mathrm{Size}\\
    \hline
    0000000000000000 & 1000000000000000 & 1 & 512\\
    1000000000000000 & 1100000000000000 & 1 & 7680\\
    1100000000000000 & 1110000000000000 & 1 & 53760\\
    1110000000000000 & 1111000000000000 & 1 & 17920\\
    1110000000000000 & 1110100000000000 & 1 & 215040\\
    1110100000000000 & 1110100010000000 & 1 & 215040\\
    0000000000000000 & 1100000000000000 & 2 & 3840\\
    1100000000000000 & 1111000000000000 & 2 & 26880\\
    1100000000000000 & 1110100000000000 & 2 & 322560\\
    1110000000000000 & 1111100000000000 & 2 & 107520\\
    1110000000000000 & 1110110000000000 & 2 & 161280\\
    1111000000000000 & 1110100000000000 & 2 & 107520\\
    1110100000000000 & 1110100011000000 & 2 & 322560\\
    0000000000000000 & 1110000000000000 & 3 & 17920\\
    1100000000000000 & 1111100000000000 & 3 & 322560\\
    1100000000000000 & 1110101000000000 & 3 & 215040\\
    1110000000000000 & 1110100010001000 & 3 & 860160\\
    0000000000000000 & 1111000000000000 & 4 & 4480\\
    0000000000000000 & 1110100000000000 & 4 & 53760\\
    0000000000000000 & 1111100000000000 & 5 & 53760\\
    0000000000000000 & 1111111100000000 & 8 & 480
  \end{array}$
  \caption{The relevant edge orbits for \hadkttlint[4]{}
  gadget. The edges of a \hadkttlint[4]{} gadget has a
  total of 1061 edge orbits, but only 21 are ever used in our constructed
  gadgets.
  Note that as discussed in Appendix \ref{sec:rest}, this list of edges was identified using the Gurobi LP-solver, and not using a rational LP solver. See Appendix \ref{sec:rest} for more information.
  \label{fig13}}
\end{table}

\end{document}